\documentclass[runningheads]{llncs}

\usepackage{amsmath,amssymb}
\usepackage{array,booktabs,float,graphicx,listings,multirow,xcolor,tikz,url}
\usepackage{longtable,capt-of,needspace}
\usetikzlibrary{arrows.meta,backgrounds,fit}
\definecolor{vanilla}{RGB}{75,105,145}
\definecolor{flexible}{RGB}{0,125,120}
\definecolor{slotblue}{RGB}{20,85,155}
\definecolor{slotorange}{RGB}{180,80,10}

\newcommand{\Seq}{\texttt{Seq}}
\newcommand{\Bag}{\texttt{Bag}}
\newcommand{\Set}{\texttt{Set}}
\newcommand{\TEmb}{\operatorname{TEmb}}
\newcommand{\TRen}{\operatorname{TRen}}
\newcommand{\EqCert}{\operatorname{EqCert}}

\providecommand{\leanref}[2]{}

\title{Typed Flexible-Arity Slotted E-Graphs:\\
A Soundness Construction and an Alloy Case Study}
\titlerunning{Typed Flexible-Arity Slotted E-Graphs}
\author{Guanxuan Wu\orcidID{0009-0006-8102-7329} \and
Allison Sullivan\orcidID{0000-0001-7400-2218}}
\authorrunning{G. Wu and A. Sullivan}
\institute{University of Texas at Arlington, Arlington, TX 76019, USA\\
\email{gxw6804@mavs.uta.edu, allison.sullivan@uta.edu}}

\begin{document}
\maketitle
\begin{abstract}
Slotted e-graphs represent open terms modulo consistent renaming, while
algebraic operators benefit from canonical sequence, bag, or set children.
In this paper, we compose the two at a specification level: typed slot-mapped invocations inhabit operator-declared
ports whose sibling quotient and recursive flattening licenses are certified
separately.  A generic finite quotient presentation proves exactness of its
least-orbit normal form, while certified records specify effective-support
kernel extraction and collision.  For abstract obligation traces carrying
local endpoint certificates, we prove finite-unfolding equational soundness.
We use an Alloy case study to compare seven
related pipeline arms on a frozen corpus and a controlled transformation
suite.  
\end{abstract}

\section{Introduction}

E-graphs provide a compact representation of many equivalent expressions, enabling equality saturation to explore alternative rewrites without discarding earlier forms and to defer the selection of a preferred expression to an application-specific cost model~\cite{ogEgraphs1,egraphsOG,eqsat,eqsat-lmcs,egg}. 
Orthogonal extensions complicate those keys: algebraic operators benefit
from sequence-, bag-, or set-valued children rather than explicit
associativity, commutativity, and idempotency (ACI) saturation; open terms
benefit from interfaces quotiented by consistent renaming.  Neither dimension
alone covers the case where an algebraic container holds a typed invocation
of an open e-class.

We give a proof-carrying composition of the typed slots and quotient space over the aforementioned algebraic properties in equality saturation. 
Each operator declares ports generated
by \texttt{One}, $K^J$ for $K\in\{\Seq,\Bag,\Set\}$, unary binding, and
binder blocks.  A child invocation $m*a$ embeds the exposed slots of e-class
$a$ into its caller.  The carrier determines only the certified sibling
quotient; a separate, operator- and path-specific certificate licenses
recursive same-head flattening.  Equality-certified permutations act on
exposed slots, while descriptor-certified automorphisms act locally inside a
binder block.  Rebuilding normalizes these actions jointly with union-find
leaders in the intended design without inferring child equality from parent
equality.  The checked trace theorem instead consumes local obligation
evidence and does not verify an operational rebuild implementation.

The proposed construction is quotient-first.  Leader finding may contract
support through a proper embedding, so canonicalization first extracts an
exact-support leader kernel and retains a source-to-kernel certificate. Such process 
then minimizes local leader and binder orbits before aggregating bags or
deduplicating sets, and only afterward minimizes the finite free-slot-renaming
orbit.  The checked mathematical layer proves exactness for any supplied
finite quotient presentation and validates supplied extraction and collision
records; showing that the concrete procedure constructs those records is a
separate refinement obligation.  A collision may create a parent edge only
when the composed certificate inhabits that edge's exact directed type.

We prove finite-unfolding soundness for abstract finite certified traces.
Their nine-step labels are classified as local obligation addition, rekeying,
composition, or removal, each carrying the needed typed endpoint evidence.
Under these premises, two certified unfoldings whose parent paths reach one
leader and whose caller maps align by a recorded symmetry word are equal in
every supplied model. 
To note, the result does not establish concrete transition
semantics, unrestricted ACI completeness, or refinement of the Java code.

To evaluate our e-graph construction on a concrete formal dataset, we explore an Alloy case study with seven related pipeline arms on a
frozen corpus and a targeted capability suite. The pipeline separates a targeted
binder-permutation capability from whole-pipeline structural consolidation;
the measurements are bounded implementation evidence, not a proof that
the adapter realizes the abstract construction. In our controlled benchmark, the slot-aware configurations recognized all 5,500 constructed equivalences, including all 500 binder-block permutations, compared with 22 for the De Bruijn baselines. Across 61,598 AST-distinct Alloy predicate pairs, the Certificate-Integrated IR recognized 4,088 correct student–oracle equivalences versus 2,160 for the De Bruijn baselines, with no observed zero-distance matches among pairs labeled incorrect.

Our contributions are: (i) a typed flexible-arity port grammar for slotted
invocations; (ii) a generic finite-orbit exactness theorem and certified
effective-support record interfaces with explicit parent-versus-interface
restriction; (iii) an abstract certificate-preserving trace contract and
finite-unfolding soundness theorem; and
(iv) a bounded Alloy study whose evidence boundary is stated separately from
the mathematical result.  Full statements and proof details appear in
Appendix~\ref{app:formal-details}. 

\section{Background and Scope}
\label{sec:background}

\paragraph{E-graphs and equality saturation.}
An e-graph compactly records expressions that have been shown equal under a
chosen set of equations. An \emph{e-node} stores an operator and references to
its child \emph{e-classes}; an e-class groups alternative expressions already
known to be equal. Thus a child reference can stand for several expressions,
and shared subexpressions need not be copied~\cite{ogEgraphs1,egraphsOG,egg}.

For instance, assume an arithmetic language in which $x\times 2=x+x$ is a valid
equation. Let $c_a=\{a\}$ and $c_2=\{2\}$ be singleton e-classes.
The term $a\times 2$ is represented by the e-node
$\times(c_a,c_2)$. Applying the equation does not replace this node:
it adds $+(c_a,c_a)$ and places both alternatives in one e-class,
$$
c=\{\times(c_a,c_2),\; +(c_a,c_a)\}.
$$
A parent e-node $g(c)$ therefore represents both $g(a\times 2)$ and
$g(a+a)$ while sharing their common subterms, as shown in Figure~\ref{fig:typed-slotted-egraphs}(a). 
\begin{figure}[t]
  \centering
  \includegraphics[width=\linewidth]{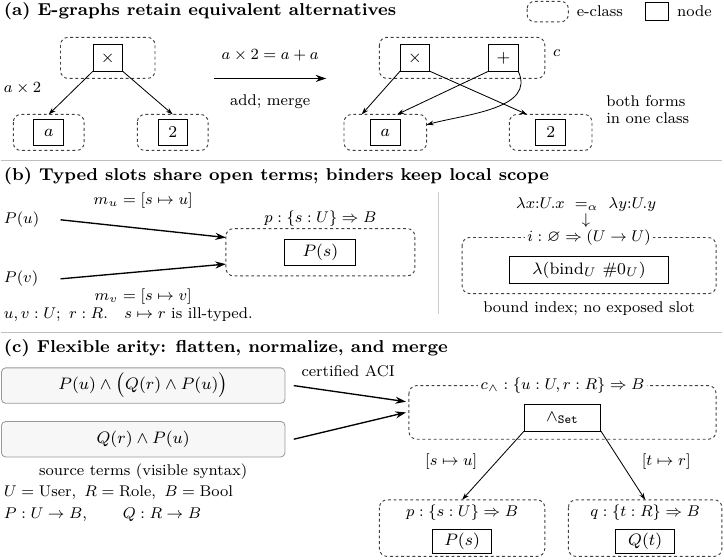}
  \caption{E-graphs, typed slots, and flexible arity.
  (a) Rewriting retains equivalent alternatives.
  (b) Typed maps share shapes without equating distinct free invocations;
  bound indices stay local.
  (c) Certified associativity flattens visible syntax; the Set quotient removes
  order and duplicates. Dashed boxes are e-classes; solid inner boxes are e-nodes.}
  \label{fig:typed-slotted-egraphs}
\end{figure}

This propagation from equal children to equal parent expressions is called 
\emph{congruence}. 
An implementation of it uses union--find to track merged
classes and \emph{hash-consing} to reuse nodes with the same operator and
child-class key. After a merge, a \emph{rebuilding} protocol updates parent keys and
merges any newly identical parents. To achieve \emph{equality saturation,} the
equation application and rebuilding process is repeated until no new nodes or equalities are found or a
resource limit is reached. The alternatives remain available until
\emph{extraction} selects a represented term, for example one with least
cost under a chosen cost model. In this paper, the e-graph provides the
shared structure on which flexible-arity operators and typed slot renamings
are combined.

\paragraph{Algebraic structure and binding.}
Ordinary positional nodes do not internalize associativity, commutativity,
and idempotence (ACI): explicit bidirectional rules may materialize many
parenthesizations and permutations. For conjunction, associativity permits
flattening nested conjunctions, commutativity forgets operand order, and
idempotence removes repeated operands. These are distinct choices: a
sequence preserves order and multiplicity, a bag forgets order but retains
multiplicity, and a set forgets both. Rewriting modulo AC and canonized
completion provides the established basis for flat bags and
sets~\cite{conchon2012canonized,kapur2021ModularAC}. Here we combine these
techniques with typed open-term interfaces.

Variable binding introduces a separate source of syntactic variation:
$\forall x{:}T.\,P(x)$ and $\forall y{:}T.\,P(y)$ differ only in a
consistently renamed bound variable. This is $\alpha$-equivalence, as a typed example is
shown in Figure~\ref{fig:typed-slotted-egraphs}(b). An open
term, such as the body $P(x)$ considered outside its quantifier, must retain
its dependence on $x$; sharing it in another scope therefore requires an
explicit account of which variable fills that position.

Prior work addresses these representation choices in several ways.
Egglog supports typed container values in a relational fixpoint
system~\cite{egglog,egglog_software}; E-graphs Modulo Theories supplies a
general interface to theory-specific canonical
domains~\cite{zucker2025omeletsneedonionsegraphs}. Lifting E-Graphs represents open
terms as functions over interfaces~\cite{zucker2026lifting}, while
Categorical E-Graphs treat binding through hierarchical
hypergraphs~\cite{egBindings}. Slotted E-Graphs instead expose a finite slot interface
$S_a$ for each class and represent an occurrence as $m*a$, where $m$ maps
$S_a$ into the caller~\cite{slottedEGraphs}. The slots record the class's
exposed variable positions; the map says which caller variables fill them.
Given a formal context with two defined types, \texttt{User} and \texttt{Photo}, 
there could be an arbitrary number of slots $u_1,...,u_k\in\texttt{User}, p_1,...,p_r\in\texttt{Photo}$
defined; 
in our typed setting, a \texttt{User} slot must map to a \texttt{User}
coordinate, and a \texttt{Photo} slot to a \texttt{Photo} coordinate.
The notation $m*a$ denotes an occurrence with that interface map, not a
claim that arbitrary assignments to its free variables give equal values.
Shapes quotient concrete slot
names; separately justified groups record symmetries, and a coordinate may be
removed only after independence is proved. De Bruijn indices normalize
ordinary alpha-renaming but do not express side-conditioned
permutations of declaration blocks~\cite{deBruijn,maziarz2021HashingAlpha}.

\paragraph{Alloy models and checks.}
Alloy is a typed relational specification language with temporal
operators~\cite{alloy,kodkod}. A model describes sets of objects and relations
between them. A \texttt{sig} declares a set of atoms, a field declares a
relation, \texttt{extends} declares a subsignature, and a \texttt{pred}
packages a formula. In the social-media model of
Figure~\ref{fig:alloy-certificate-example}, \texttt{User} and \texttt{Photo}
are sets, \texttt{Ad} is a subset of \texttt{Photo}, and
\texttt{follows: set User} associates each user with any number of users.
The multiplicity \texttt{one} in \texttt{date: one Day} requires exactly
one day per photo. Dot denotes relational join:
\texttt{u.follows.posts} is the set of photos posted by users whom
\texttt{u} follows. The formula \texttt{p in u.sees} tests whether
\texttt{u} sees \texttt{p}, and \texttt{all u: User |} universally
quantifies a user.

The Alloy Analyzer searches for instances satisfying a predicate with
\texttt{run}, or for counterexamples to an assertion with \texttt{check}.
The search uses a specified object scope and, for bounded temporal analysis,
trace bounds. Checking equivalence asks whether the two predicates can
disagree in a permitted instance; finding no disagreement is evidence within
those bounds. In our dataset, each exercise supplies an oracle predicate
expressing the intended constraint. A student predicate is overconstrained
when it excludes oracle-allowed instances, underconstrained when it admits
oracle-forbidden instances, or both. The evaluation inherits the dataset's
labels and reports its own bounded checks separately.

\paragraph{Scope of the construction.}
Our scope is the following composition: operator-declared
$\Seq/\Bag/\Set$ ports contain typed slot-mapped invocations, and their keys
are rebuilt jointly with leaders, invocation embeddings, certified slot
symmetries, and binder-block automorphisms. A port is an operator's declared
argument position or argument container. Flat construction opens only
visible applications of the same type-instantiated operator; it never unfolds
an opaque invocation. A carrier annotation contributes no semantic equation
without its endpoint-indexed certificate: the evidence must identify the
exact two well-typed expressions that the step equates.

In the Alloy study, \emph{canonical} means canonical under the implementation's
stored keys and declared rewrite vocabulary; the checked normal-form result
is instead generic over a supplied finite presentation. Neither layer claims
semantic completeness, unrestricted saturation termination, parser
correctness, full Java refinement, or full-corpus experimental replay.

\section{Typed Flexible-Arity Slotted E-Graphs}
\label{sec:formalism}

We fix one typed source signature and semantic profile.  Every quotient,
flattening law, binder automorphism, and graph equality used below carries an
endpoint-indexed certificate; annotations alone introduce no equation.

\subsection{Typed contexts, embeddings, and ports}

The familiar cost of materializing AC closure motivates a flat carrier.
\begin{lemma}[Size of a Complete Binary AC Closure]
\label{res:F01}
\leanref{TSG-FND-001}{TypedSlottedEGraphsPaper.ACClosureCounts.completeBinaryACClosureCounts}
\leanref{TSG-ATOM-001}{TypedSlottedEGraphsPaper.ACClosureCounts.completeBinaryACProductClassCount}
\leanref{TSG-ATOM-002}{TypedSlottedEGraphsPaper.ACClosureCounts.completeBinaryACOperatorNodeCount}
Represent a product class in the free commutative semigroup by a nonempty
square-free subset of $n$ generators, and a binary operator node by an ordered
pair of disjoint nonempty subsets.  For every $n$, these two finite types have
\leanref{TSG-ATOM-043}{TypedSlottedEGraphsPaper.ACClosureCounts.productClassCountThetaTwoPow}
\leanref{TSG-ATOM-044}{TypedSlottedEGraphsPaper.ACClosureCounts.binaryOperatorNodeCountThetaThreePow}
constructive codings of sizes $2^n-1$ and $3^n+1-2^{n+1}$, respectively.
Their counting functions are $\Theta(2^n)$ and $\Theta(3^n)$.
\end{lemma}
\begin{proof}See Appendix~\ref{app:proof-F01}.\end{proof}

\leanref{TSG-FND-004}{TypedSlottedEGraphsPaper.AtomicLedger.FiniteContainerCarrierDefinition}
For a carrier $X$, let
\begin{align*}
\Seq(X)&=\coprod_{k\in\mathbb N}X^k,&
\Bag(X)&=\{\mu:X\to\mathbb N\mid\operatorname{supp}(\mu)\text{ finite}\},\\
\Set(X)&=\{Y\subseteq X\mid |Y|<\infty\},&
K^J(X)&=\{c\in K(X)\mid |c|\in J\}.
\end{align*}
Here bag cardinality counts occurrences, $J\ne\varnothing$, and
$K\in\{\Seq,\Bag,\Set\}$; write $K^+$, $K^0$, and $K^{=k}$ for
$J=\mathbb N_{>0},\mathbb N,\{k\}$.  Sequences observe order and
multiplicity, bags quotient order, and sets quotient order and duplication.

\begin{definition}
\label{res:F02}
\leanref{TSG-DEF-001}{TypedSlottedEGraphsPaper.FiniteTypedContext.f02CanonicalContextEncoding}
Fix an injective rank $\operatorname{rank}:\mathsf{Ty}\to\mathbb N$; types
may, for example, be generated by
\[
\tau::=\alpha\mid\texttt{Int}\mid\texttt{Bool}\mid\tau_1\to\tau_2
\mid\texttt{Rel}(\tau_1,\ldots,\tau_k)\mid T(\tau_1,\ldots,\tau_k).
\]
A \textbf{finite typed context} is a finite list of coordinate types;
positions are distinct slots even when their types agree.  Its canonical
context $\operatorname{Can}(\Gamma)$ is the stable rank-sorted list.  It
preserves every type-fibre cardinality and names each sorted coordinate by
$(\tau,\mathsf{free},i)$, where $i$ is below that fibre's cardinality.  Bound
coordinates are represented separately by a typed de Bruijn scope.
\end{definition}

\begin{definition}
\label{res:F03}
\leanref{TSG-DEF-002}{TypedSlottedEGraphsPaper.paperTypedEmbeddingRenamingDefinition}
For finite typed contexts $S,T$, a \textbf{typed slot embedding} is a
type-preserving injection,
\[
\TEmb(S,T)=\{m:S\hookrightarrow T\mid
\operatorname{ty}(m(x))=\operatorname{ty}(x)\}.
\]
A \textbf{typed slot renaming} is an embedding equipped with a typed inverse
and both inverse laws.  Write $\TRen(S,T)$ for these proof-relevant
isomorphisms; a typed permutation is an element of $\TRen(S,S)$.
\end{definition}

An e-class identifier $a$ has interface $(S_a,\tau_a)$.  Its invocation in
caller context $\Gamma$ is $m*a$ with $m\in\TEmb(S_a,\Gamma)$ and output
$\tau_a$; $a$ abbreviates $\operatorname{id}_{S_a}*a$.  The star is a
constructor, not a map action.

\begin{proposition}
\label{res:F04}
\leanref{TSG-PROP-001}{TypedSlottedEGraphsPaper.typedEmbeddingCategoryLaws}
\leanref{TSG-ATOM-003}{TypedSlottedEGraphsPaper.TypedSlotEmbedding.id}
\leanref{TSG-ATOM-004}{TypedSlottedEGraphsPaper.TypedSlotEmbedding.comp}
For finite typed contexts $S,T,T'$, the identity is injective and composition
preserves typed embeddings:
\[
\operatorname{id}_S\in\TEmb(S,S),\qquad
\frac{m\in\TEmb(S,T)\quad e\in\TEmb(T,T')}{e\circ m\in\TEmb(S,T')}.
\]
\leanref{TSG-ATOM-005}{TypedSlottedEGraphsPaper.TypedSlotRenaming.symm}
For $\rho\in\TRen(S,T)$, its stored typed inverse lies in $\TRen(T,S)$ and
$\rho^{-1}\circ\rho=\operatorname{id}_S$.
\end{proposition}
\begin{proof}Directly from injection, surjection, and type preservation; see
Appendix~\ref{app:proof-F04}.\end{proof}

\leanref{TSG-FND-006}{TypedSlottedEGraphsPaper.ClaimLedger.TSGFND006BinderDescriptorOccurrenceBundle}
Each binder descriptor $\beta$ has a finite typed context $\Delta_\beta$,
domain, multiplicity, disjointness, and dependency data, and a certified subgroup
$\operatorname{Aut}(\beta)\leq\operatorname{Perm}_{\rm ty}(\Delta_\beta)$.
An occurrence in $\Gamma$ is $o=(\Delta_o,a_o)$ with
$\Delta_o\cap\Gamma=\varnothing$ and
$a_o\in\TRen(\Delta_\beta,\Delta_o)$; write
$\operatorname{BOcc}_\beta(\Gamma)$ for all such occurrences.

\begin{definition}
\label{res:F05}
\leanref{TSG-DEF-003}{TypedSlottedEGraphsPaper.TypedPortSchemaAndValueDefinition}
A schema is generated by $\texttt{One}(\tau)$, an arity-licensed tagged
$\Seq/\Bag/\Set$ container, $\texttt{Bind}(\tau,\kappa)$, or
$\texttt{BindBlock}(\beta,\kappa)$.  A checked port value is indexed by its
free context, typed de Bruijn bound scope, and schema.  Its atoms are typed
free coordinates, typed bound indices, or invocations $m*a$.  A container
stores a length-indexed occurrence vector and a proof that its length is
admitted.  Unary and block binders extend the bound-type list by one type or
by the descriptor's type list.  Bag/set quotient behavior is not definitional
syntax equality: the extensional finite-bag and finite-set carriers above are
separate mathematical definitions, and their matching rules enter the
structural relation below.
\end{definition}

\leanref{TSG-FND-007}{TypedSlottedEGraphsPaper.AtomicLedger.SignaturePathLawful}
For every container-valued port path $p$ of an exact operator instance
$f_\theta$, the signature records
$\operatorname{Law}_\Sigma(f_\theta,p)\subseteq
\{\mathsf A,\mathsf C,\mathsf I,\mathsf U\}$.  A $\Seq$, $\Bag$, or $\Set$
port respectively requires $\varnothing$, $\{\mathsf C\}$, or
$\{\mathsf C,\mathsf I\}$.  Set arities and explicitly flat arities satisfy
\leanref{TSG-INV-001}{TypedSlottedEGraphsPaper.AtomicLedger.SetArityDownwardClosed}
\begin{align}
k\in J, k>0, 1\le j\le k&\Longrightarrow j\in J,
\label{eq:set-arity-closure}\\
\leanref{TSG-INV-002}{TypedSlottedEGraphsPaper.AtomicLedger.FlatAritySpliceClosed}
k,l\in J, k>0&\Longrightarrow k+l-1\in J.
\label{eq:flat-arity-closure}
\end{align}
\leanref{TSG-FND-008}{TypedSlottedEGraphsPaper.AtomicLedger.FlatLicense}
The separate license $\operatorname{Flat}_\Sigma(f_\theta,p)$ is admitted
only for a sole root container whose element and result types agree, and
requires $\mathsf A$; if $0\in J$, it also requires a source unit and
$\mathsf U$.  Thus a carrier never licenses flattening by itself.

\leanref{TSG-FND-009}{TypedSlottedEGraphsPaper.AtomicLedger.CIAEndpointCertificateFamilies}
Let $\mathcal B_\Sigma$ be an independently fixed typed source proof theory
for the profile.  For
$\mathfrak L=(\operatorname{profile}(\Sigma),f_\theta,p,K,J)$ and fixed
outside arguments, write $F(\bar u)$ for the declared source realizer.
The labels denote families of inhabitants of $\EqCert_{\mathcal B_\Sigma}$
with the following endpoints:
\begin{align*}
\mathsf C&:\ F(u_1,\ldots,u_k)=F(u_{\varpi(1)},\ldots,u_{\varpi(k)}),\\
\mathsf I&:\ F(v_{s(1)},\ldots,v_{s(k)})=F(v_1,\ldots,v_j),\\
\mathsf A&:\ F(\bar u_{<i},F(\bar v),\bar u_{>i})
=F(\bar u_{<i},\bar v,\bar u_{>i}),
\end{align*}
for every well-typed context and legal permutation, typed surjection, or
\leanref{TSG-FND-010}{TypedSlottedEGraphsPaper.AtomicLedger.UnitEndpointCertificateFamily}
splice.  When $0\in J$, $\mathsf U$ additionally certifies
$F()=e_{\mathfrak L}$ and deletion of an empty nested occurrence.  Every
concrete endpoint is checked; the label is not evidence.

\begin{definition}[Typed Flexible-Arity E-Nodes]
\label{res:F06}
\leanref{TSG-DEF-004}{TypedSlottedEGraphsPaper.TypedENodeAndSupportDefinition}
A concrete typed signature supplies class interfaces and outputs,
descriptor-bound type lists, and each already-instantiated operator's input
schemas and output type.  A typed node in $\Gamma$ is
\[
n=f(q_1,\ldots,q_r),\qquad
q_i\in\operatorname{Port}_\Gamma(\kappa_i),
\quad\operatorname{out}(n)=\tau_f.
\]
Write $\operatorname{Node}_\Gamma(\tau')$ for all such nodes of output
$\tau'$.  Support is defined recursively by
\begin{align*}
\operatorname{slots}(x)&=\{x\},&
\operatorname{slots}(m*a)&=\operatorname{im}(m),\\
\operatorname{slots}(c)&=\bigcup_{q\text{ in }c}\operatorname{slots}(q),&
\operatorname{slots}(f_\theta(\bar q))&=\bigcup_i\operatorname{slots}(q_i),\\
\operatorname{slots}(\text{bound index})&=\varnothing,&
\operatorname{slots}(\operatorname{bind}\ q)&=\operatorname{slots}(q),\\
\operatorname{slots}(\operatorname{bindBlock}_\beta q)&=\operatorname{slots}(q).&&
\end{align*}
\end{definition}

\leanref{TSG-DEF-005}{TypedSlottedEGraphsPaper.ClaimLedger.TSGDEF005ConcreteEmbeddingActionBundle}
Every $e\in\TEmb(\Gamma,\Gamma')$ acts on the free-context index of ports and
nodes: $e\cdot x=e(x)$, $e\cdot(m*a)=(e\circ m)*a$, containers act
componentwise, typed de Bruijn indices are fixed, and both binder constructors
recurse under their extended bound-type lists.  This action chooses no fresh
free coordinate.

\begin{lemma}[Typed Embeddings Preserve Types]
\label{res:F07}
\leanref{TSG-LEM-002}{TypedSlottedEGraphsPaper.concreteTypedEmbeddingsPreserveTypes}
\leanref{TSG-ATOM-006}{TypedSlottedEGraphsPaper.ConcretePort.act}
\leanref{TSG-ATOM-007}{TypedSlottedEGraphsPaper.ConcreteFlexibleArityENode.act}
If $q\in\operatorname{Port}_\Gamma(\kappa)$,
$n\in\operatorname{Node}_\Gamma(\tau)$, and
$e\in\TEmb(\Gamma,\Gamma')$, then
\[
e\cdot q\in\operatorname{Port}_{\Gamma'}(\kappa),\qquad
e\cdot n\in\operatorname{Node}_{\Gamma'}(\tau).
\]
\leanref{TSG-ATOM-008}{TypedSlottedEGraphsPaper.AtomicLedger.embeddedNodeOutputTypeUnchanged}
In particular, $\operatorname{out}(e\cdot n)=\operatorname{out}(n)$.
\end{lemma}
\begin{proof}See Appendix~\ref{app:proof-F07}.\end{proof}

\begin{corollary}[Slot Support Invariance]
\label{res:F08}
\leanref{TSG-COR-001}{TypedSlottedEGraphsPaper.concreteSlotSupportInvariance}
\leanref{TSG-ATOM-009}{TypedSlottedEGraphsPaper.concretePortSupportImageEquality}
\leanref{TSG-ATOM-010}{TypedSlottedEGraphsPaper.concreteNodeSupportImageEquality}
Under the hypotheses above,
$\operatorname{slots}(e\cdot q)=e[\operatorname{slots}(q)]$ and
$\operatorname{slots}(e\cdot n)=e[\operatorname{slots}(n)]$.
\end{corollary}
\begin{proof}By the induction for Lemma~\ref{res:F07}; see
Appendix~\ref{app:proof-F08}.\end{proof}

\begin{corollary}[Type Safety of Invocation]
\label{res:F09}
\leanref{TSG-COR-002}{TypedSlottedEGraphsPaper.concreteTypeSafetyOfInvocation}
\leanref{TSG-ATOM-011}{TypedSlottedEGraphsPaper.concreteInvocation}
If $m\in\TEmb(S_a,\Gamma)$, then $m*a$ has type $\tau_a$ in $\Gamma$.
\leanref{TSG-ATOM-012}{TypedSlottedEGraphsPaper.concreteIdentityInvocation}
\leanref{TSG-ATOM-013}{TypedSlottedEGraphsPaper.AtomicLedger.invocationActionComposition}
Moreover, $\operatorname{id}_{S_a}*a=a$, and for every
$e\in\TEmb(\Gamma,\Gamma')$, $e\cdot(m*a)=(e\circ m)*a$.
\end{corollary}
\begin{proof}Immediate from the invocation constructor and action; see
Appendix~\ref{app:proof-F09}.\end{proof}

\subsection{Graph state and graph-relative structure}

\leanref{TSG-DEF-006}{TypedSlottedEGraphsPaper.ClaimLedger.TSGDEF006TypedShapeCarrier}
The raw typed-shape carrier is a well-typed node over a canonical finite free
context.  Its sequence, bag, and set tags all retain finite occurrence vectors;
quotient normalization and set deduplication are separate supplied operations.

\begin{definition}[Flexible Arity Typed Slotted E-Graph]
\label{res:F10}
\leanref{TSG-DEF-007}{TypedSlottedEGraphsPaper.QuiescentFiniteTypedEGraph.f10FiniteQuiescentGraphEncoding}
A flexible-arity typed slotted e-graph environment packages a finitely encoded class carrier and:
\begin{itemize}
\item class outputs and interfaces, finite stored-shape domains, and typed
ambient records $B_a(p)=(T_{a,p},\rho_{a,p})$ satisfying
\[
p\in\operatorname{Shape}_{\tau_a},\quad S_a\subseteq T_{a,p},\quad
\rho_{a,p}\in\TRen(\operatorname{slots}(p),T_{a,p});
\]
\item a depth-decreasing typed parent forest (roots have no parent) and a
total recursive $\operatorname{find}$; and
\item a collision-owner relation and per-class symmetry predicates closed
under identity, inverse, and composition.
\end{itemize}
\leanref{TSG-INV-003}{TypedSlottedEGraphsPaper.QuiescentFiniteTypedEGraph.F10FiniteQuiescentGraphEnvironment.collisionIndexExact}
The checked depth makes $\operatorname{find}$ total.  Representing $H$ by its
membership relation, $a\in H(p)$ exactly when $a$ is a leader and
$p\in\operatorname{dom}(B_a)$.  Thus a bucket may contain several leaders and
asserts only collision candidacy; it has no owners when that set is empty.
\end{definition}

\begin{lemma}[Leaderization Contracts Support]
\label{lem:find-support-contracts}
\leanref{TSG-LEM-003}{TypedSlottedEGraphsPaper.leaderizationContractsSupport}
If the stored parent path for $a$ composes with a caller embedding $m$, then
every slot in the composite image already lies in $\operatorname{im}(m)$.
Thus finding a leader cannot add caller support to an invocation, although a
proper parent embedding may contract that invocation's support.
\end{lemma}
\begin{proof}See Appendix~\ref{app:proof-F11}.\end{proof}

\leanref{TSG-CE-002}{TypedSlottedEGraphsPaper.AtomicLedger.properEmbeddingDoesNotImplyIndependence}
For every Boolean value supplied by an empty-context witness, it is false that
every one-slot Boolean input equals that value; thus no such value realizes the one-slot identity.
\leanref{TSG-INV-004}{TypedSlottedEGraphsPaper.CertifiedParentEdge}
In the graph contract, a proper parent embedding is not treated as a renaming or as evidence of independence.
It is admissible only with the exact parent-edge equation (PC) below.  Changing the stored interface is a
distinct abstract rekey/restriction obligation with checked endpoint bridges; no implementation-level refactoring follows.

\begin{definition}[Typed $\alpha$-Equivalence of Port Values]
\label{def:alpha-port}
\leanref{TSG-DEF-008}{TypedSlottedEGraphsPaper.ConcreteTypedAlphaPort}
For $q\in\operatorname{Port}_{\Gamma}(\kappa)$,
$r\in\operatorname{Port}_{\Gamma'}(\kappa)$, and
$\rho\in\TRen(\Gamma,\Gamma')$, the indexed judgement
$q\equiv_{\alpha;\kappa}^{\rho}r$ is the reflexive--symmetric--transitive
closure, indexed by composition of renamings, of the following direct
structural generator.
\begin{itemize}
\item Slots require $\rho(x)=y$; invocations require the same identifier and
$m'=\rho\circ m$.
\item Sequences relate pointwise at equal length.  Bags relate when a
bijection between expanded occurrence tokens pairs related values, thereby
preserving multiplicity.  Sets relate when every element on either
side has a related mate on the other.
\item A unary binder fixes the new head de Bruijn variable and acts on the
outer tail by the selected typed permutation.
\item A block binder chooses a typed permutation of the descriptor prefix,
certified by the signature's declared automorphism policy, and acts on the
outer tail by the selected typed permutation.
\end{itemize}
\end{definition}

\begin{proposition}[Typed Alpha-Equivalence Laws]
\label{res:F13}
\leanref{TSG-PROP-002}{TypedSlottedEGraphsPaper.concreteTypedAlphaFullEquivalenceLaws}
Let $q\in\operatorname{Port}_\Gamma(\kappa)$,
$r\in\operatorname{Port}_{\Gamma'}(\kappa)$, and
$s\in\operatorname{Port}_{\Gamma''}(\kappa)$, with
$\rho_1\in\TRen(\Gamma,\Gamma')$ and
$\rho_2\in\TRen(\Gamma',\Gamma'')$.  Then
\leanref{TSG-ATOM-014}{TypedSlottedEGraphsPaper.ConcreteAlphaPortClosure.reflexive}
\begin{align*}
q&\equiv_{\alpha;\kappa}^{\operatorname{id}_\Gamma}q,\\
\leanref{TSG-ATOM-015}{TypedSlottedEGraphsPaper.ConcreteAlphaPortClosure.symmetric}
q\equiv_{\alpha;\kappa}^{\rho_1}r
&\Longrightarrow r\equiv_{\alpha;\kappa}^{\rho_1^{-1}}q,\\
\leanref{TSG-ATOM-016}{TypedSlottedEGraphsPaper.ConcreteAlphaPortClosure.transitive}
q\equiv_{\alpha;\kappa}^{\rho_1}r\land
r\equiv_{\alpha;\kappa}^{\rho_2}s
&\Longrightarrow q\equiv_{\alpha;\kappa}^{\rho_2\circ\rho_1}s.
\end{align*}
\leanref{TSG-ATOM-017}{TypedSlottedEGraphsPaper.concreteUnindexedAlphaPortEquivalenceLaws}
Consequently, the dependent-pair unindexed relation
$q\equiv_{\alpha;\kappa}r\iff
\exists\rho.\ q\equiv_{\alpha;\kappa}^{\rho}r$ is an equivalence relation.
\end{proposition}
\begin{proof}See Appendix~\ref{app:proof-F13}.\end{proof}

\leanref{TSG-DEF-009}{TypedSlottedEGraphsPaper.QuiescentFiniteTypedEGraph.GraphStructuredAlignedBy}
The graph-relative relation $\equiv_{\mathcal G;\alpha}$ changes only the
invocation clause.  For appropriately typed $m*a,m'*b$ and
$\rho\in\TRen(\Gamma,\Gamma')$,
\begin{align*}
&m*a\equiv_{\mathcal G;\alpha;\texttt{One}(\tau)}^{\rho}m'*b\\
&\quad\Longleftrightarrow
\exists\ell,\widehat m,\widehat m',\pi.\ 
\tau_a=\tau_b=\tau\ \land\
\operatorname{find}_\Gamma(m*a)=\widehat m*\ell\ \land\\[-0.2em]
&\hspace{11em}
\operatorname{find}_{\Gamma'}(m'*b)=\widehat m'*\ell\ \land\
\pi\in G_\ell\ \land\
\rho\circ\widehat m=\widehat m'\circ\pi.
\end{align*}
Consistently with this boundary, every enclosing clause remains structural.
\leanref{TSG-CE-001}{TypedSlottedEGraphsPaper.AtomicLedger.parentEqualityDoesNotImplyChildSymmetry}
There exist a Boolean child observation, a swap, a parent observation, and an input for
which the parent observations agree while the corresponding child observations differ; thus parent equality alone does not imply child symmetry.

\begin{lemma}[Transport of Graph-Relative Invocation Alignment]
\label{lem:graph-rel-transport}
\leanref{TSG-LEM-004}{TypedSlottedEGraphsPaper.QuiescentFiniteTypedEGraph.graphInvocationAlignment_transport}
Let two invocation leaves in caller contexts $\Lambda,\Lambda'$ be aligned
under $\rho\in\TRen(\Lambda,\Lambda')$ by a common found leader and a
recorded leader symmetry.  For
$\eta\in\TRen(\Lambda,\Xi)$ and
$\eta'\in\TRen(\Lambda',\Xi')$, put
$\bar\rho=\eta'\circ\rho\circ\eta^{-1}\in\TRen(\Xi,\Xi')$.
Renaming both callers preserves that invocation alignment under $\bar\rho$.
\end{lemma}
\begin{proof}See Appendix~\ref{app:proof-F14}.\end{proof}

\leanref{TSG-DEF-010}{TypedSlottedEGraphsPaper.PaperCertifiedKernelExtraction}
\paragraph{Leader kernels.}
The intended producer replaces invocations by found leaders, composes parent
embeddings, and reapplies only declared sibling quotients and flat smart
constructors.  Its checked interface is instead a supplied extraction record:
an original term over $\Gamma_0$, a kernel over $\Delta_n$, an inclusion
$\iota_n:\Delta_n\hookrightarrow\Gamma_0$, and finite replay provenance.
The provenance is not part of a collision key; constructing this record from
a graph remains a refinement obligation.

\leanref{TSG-DEF-011}{TypedSlottedEGraphsPaper.QuiescentFiniteTypedEGraph.GraphStructuredValue.LeaderNormalized}
For quiescent $\mathcal G$, let
$\operatorname{LPort}_{\mathcal G,\Lambda}(\kappa)$ and
$\operatorname{LNode}_{\mathcal G,\Lambda}(\tau)$ be values whose invocations
all name leaders.  Their support need only be contained in $\Lambda$.

\begin{lemma}[Identity Alignment Preserves Leader-Invocation Support]
\label{lem:leader-id-support}
\leanref{TSG-LEM-005}{TypedSlottedEGraphsPaper.QuiescentFiniteTypedEGraph.identityAlignedLeaderSupport}
If two invocations of the same leader are identity-aligned in $\Lambda$
through an admitted leader symmetry, then for every caller slot their two
image-support predicates are equivalent.
\end{lemma}
\begin{proof}See Appendix~\ref{app:proof-F15}.\end{proof}

\leanref{TSG-DEF-012}{TypedSlottedEGraphsPaper.NormalizationBridge.AmbientQuotientPresentation.quotientNormal}
\paragraph{Quotient-first normal form.}
The concrete design minimizes child and block orbits before bag aggregation
or set deduplication.  The following claims are stated over a finite quotient
presentation.  A \emph{finite quotient presentation} $P$ supplies a linearly
ordered carrier $X$, local generators fixed independently of normalization,
a decision procedure for their reflexive--symmetric--transitive closure
$\mathrel{\sim_P}$, an explicit finite carrier containing every value, and
support, schema, leaf-type, and output observations preserved by every local
generator.  Define $Q_P(x)$ as the least member of the enumerated
$\sim_P$-orbit of $x$.  Establishing that a particular recursive program
constructs such a presentation is a separate refinement obligation.

\begin{lemma}[Ambient Quotient-Normal-Form Exactness]
\label{lem:qnf-exact}
\leanref{TSG-LEM-006}{TypedSlottedEGraphsPaper.NormalizationBridge.AmbientQuotientPresentation.ambientQuotientNormalFormExactness}
\leanref{TSG-ATOM-020}{TypedSlottedEGraphsPaper.NormalizationBridge.AmbientQuotientPresentation.quotientNormal_sound}
For every $q,r\in X$,
\begin{align*}
Q_P(q)&\sim_P q,\\
\leanref{TSG-ATOM-021}{TypedSlottedEGraphsPaper.NormalizationBridge.AmbientQuotientPresentation.quotientNormal_eq_iff}
Q_P(q)=Q_P(r)&\Longleftrightarrow q\sim_P r,\\
\leanref{TSG-ATOM-022}{TypedSlottedEGraphsPaper.NormalizationBridge.AmbientQuotientPresentation.quotientNormal_preserves_support}
\operatorname{support}(Q_P(q))&=\operatorname{support}(q).
\end{align*}
\leanref{TSG-ATOM-023}{TypedSlottedEGraphsPaper.NormalizationBridge.AmbientQuotientPresentation.quotientNormal_preserves_schema_leafTypes_output}
It also preserves the recorded schema, leaf types, and output, and is
\leanref{TSG-ATOM-025}{TypedSlottedEGraphsPaper.NormalizationBridge.AmbientQuotientPresentation.quotientNormal_idempotent}
idempotent: $Q_P(Q_P(q))=Q_P(q)$.
\end{lemma}
\begin{proof}See Appendix~\ref{app:proof-F16}.\end{proof}

The proposed producer therefore normalizes elements before bag aggregation or
set deduplication; this ordering is a design choice outside the generic
finite-presentation theorem.

\leanref{TSG-DEF-013}{TypedSlottedEGraphsPaper.NormalizationBridge.EffectiveSupportCanonicalizer}
\leanref{TSG-INV-005}{TypedSlottedEGraphsPaper.NormalizationBridge.EffectiveSupportCanonicalizer.canonicalRecordPostconditions}
\paragraph{Canonical typed shape.}
An \emph{effective-support canonicalizer} $C$ packages a finite structural
quotient presentation, a constructive selector for its proof-relevant typed
renaming witness, and node-to-kernel extraction with exact support, schema,
leaf-type, and output observations.  It also packages a generally
nonsurjective kernel-to-ambient inclusion and provenance.  Its canonical shape is
$\operatorname{Shape}_C(n)=Q_C(K_C(n))$.  The returned record keeps the typed
shape witness separate from ambient transport, which is definitionally the
inclusion composed with that witness and is never inverted.  A concrete
canonical-alphabet enumeration may instantiate this interface, but leastness
of a particular witness is not part of the checked theorem.

\begin{corollary}[Effective-Support Canonical-Shape Exactness]
\label{cor:shape-exactness}
\leanref{TSG-COR-003}{TypedSlottedEGraphsPaper.NormalizationBridge.EffectiveSupportCanonicalizer.effectiveSupportCanonicalShapeExactness}
For nodes $n,n'$ in an effective-support canonicalizer $C$, let their kernels
have contexts $\Delta_n,\Delta_{n'}$.  Then
\begin{align*}
\leanref{TSG-ATOM-026}{TypedSlottedEGraphsPaper.NormalizationBridge.EffectiveSupportCanonicalizer.canonicalShape_eq_iff_effectiveKernel_typedRenaming}
\operatorname{Shape}_{C}(n)=\operatorname{Shape}_{C}(n')
\quad\Longleftrightarrow\quad
&\exists\rho\in\TRen(\Delta_n,\Delta_{n'}).\\[-2pt]
&\operatorname{Nonempty}(K_C(n)\Rightarrow_\rho K_C(n')).
\end{align*}
where $\Rightarrow_\rho$ is the proof-relevant generated structural
\leanref{TSG-ATOM-027}{TypedSlottedEGraphsPaper.NormalizationBridge.EffectiveSupportCanonicalizer.canonicalShape_exactness_preserves_output}
derivation.  Either side implies equality of the recorded node outputs.  No
renaming between the larger ambient input contexts is asserted.
\end{corollary}
\begin{proof}See Appendix~\ref{app:proof-F17}.\end{proof}

\subsection{Canonicalization and rebuilding}
\label{sec:algorithm}

\leanref{TSG-ALG-002}{TypedSlottedEGraphsPaper.ClaimLedger.TSGALG002CompleteFiniteCanonicalizationContract}
The checked canonicalization contract takes a finite candidate list that is
extensionally complete in both directions, selects a present least candidate,
and returns a record whose candidate is exactly that selection.  The selected
candidate is therefore complete, and only the returned shape projection is a
collision key.  Producing that list by recursive typed-bijection and local-
orbit enumeration is a separate refinement obligation.

\leanref{TSG-ACC-002}{TypedSlottedEGraphsPaper.ClaimLedger.structuralComputationBudget}
\leanref{TSG-DEF-014}{TypedSlottedEGraphsPaper.AtomicLedger.canonicalizationWorkFactor}
Let $S_{\rm in}$ be expanded input size, $S_K$ kernel size, $L$ retained
provenance length, $\operatorname{InvOcc}$ the multiset of leader-invocation
occurrences, and $\operatorname{BlockOcc}$ the block occurrences.  To expose
the sources of combinatorial work, define the reporting factor and budget
\begin{align*}
P_Q(n)&=\left(\prod_\tau|\Delta_{n,\tau}|!\right)
\left(\prod_{(\beta,o)\in\operatorname{BlockOcc}}|\operatorname{Aut}(\beta)|\right)\\[-2pt]
&\qquad\cdot\left(1+\sum_{m*\ell\in\operatorname{InvOcc}}|G_\ell|\right),\\
\leanref{TSG-ACC-001}{TypedSlottedEGraphsPaper.ClaimLedger.ceilLog2}
\ell(s)&=\operatorname{ceilLog2}(2+s),\\
B_T(n)&=S_{\rm in}\ell(S_{\rm in})+L+
P_Q(n)S_K\ell(S_K).
\end{align*}
Here the executable natural-number function is zero at inputs $0$ and $1$.
For every $k$, it satisfies
$\ell(k)=\lfloor\log_2(k+1)\rfloor+1$.
\leanref{TSG-ACC-003}{TypedSlottedEGraphsPaper.ClaimLedger.auxiliaryKernelWorkspaceBudget}
\leanref{TSG-ACC-004}{TypedSlottedEGraphsPaper.ClaimLedger.returnedStructuralRecordBudget}
The budget uses comparison sorting and charges map operations to encoded
sizes; its intended auxiliary and returned-record budgets are $S_K$ and
$S_K+L$, excluding input, stored groups, and materialized replay certificates.
These are accounting definitions, not proved asymptotic bounds: the Lean
development contains no operational cost semantics, and no Java runtime claim
is inferred from them.

\leanref{TSG-ALG-003}{TypedSlottedEGraphsPaper.ClaimLedger.TSGALG003CompatibilityGatedCollisionContract}
\leanref{TSG-ALG-004}{TypedSlottedEGraphsPaper.ClaimLedger.TSGALG004CertificatePreservingRebuildContract}
The intended producer allocates $\Delta_n$, gates collision union on the two
source-to-kernel certificates and a directed (PC) certificate, and treats
support shrinkage as an explicit interface restriction.  These are producer
obligations: the trace theorem consumes only certified obligation additions, rekeys,
compositions, and removals, not allocation, hashing, or a dirty-set algorithm.

\begin{proposition}[Finite Rebuild Quiescence]
\label{res:F18}
\leanref{TSG-PROP-003}{TypedSlottedEGraphsPaper.finiteRebuildQuiescence}
Let a rebuild system have a natural-valued measure, a step relation, and a
quiescence predicate.  If every step strictly decreases the measure and every
nonquiescent state admits a step, then every initial state reaches a
quiescent state by a finite rebuild trace.
\end{proposition}
\begin{proof}See Appendix~\ref{app:proof-F18}.\end{proof}

\section{Certificate-Based Metatheory}
\label{sec:metatheory}

\leanref{TSG-FND-011}{TypedSlottedEGraphsPaper.StructuralLocalLaws}
The generic local-law interface declares descriptor-local automorphism tokens,
but a token alone proves no equation.  Its block-automorphism rule accepts an
already checked certificate between the selected body endpoints and lifts that
certificate under the matching binder block.  Constructing such endpoints or
their plug/conjugation certificate from concrete descriptor metadata is a
separate refinement obligation.
\leanref{TSG-DEF-015}{TypedSlottedEGraphsPaper.TypedTermLanguage}
A \texttt{TypedTermLanguage} supplies its typed term family and renaming action
together with primitive-equation and forward structural-congruence relations.
The intended source instance may populate those relations with checked alpha,
descriptor, $\mathsf A/\mathsf C/\mathsf I/\mathsf U$, and input-equation
endpoints; the generic Lean theorems assume rather than construct that
instantiation.  Neither a schema annotation nor operator injectivity is assumed.

\begin{definition}[Typed Equational Certificate]
\label{def:eq-derivation}
\leanref{TSG-DEF-016}{TypedSlottedEGraphsPaper.TypedEquationalCertificate}
For $s,t\in\operatorname{Term}_\Gamma(\tau)$,
$\EqCert_{\mathcal T}(\Gamma;s,t)$ is the type of finite well-typed
equational derivation trees concluding $\mathcal T\vdash_\Gamma s=t$.  Its
constructors are reflexivity, a primitive equation, a forward structural-
congruence rule, symmetry, transitivity, and transport along a typed
embedding.  In particular,
\[
d\in\EqCert_{\mathcal T}(\Gamma;s,t),\quad
e\in\TEmb(\Gamma,\Gamma')
\]
induces
\[
e_*d\in\EqCert_{\mathcal T}
(\Gamma';e\cdot s,e\cdot t).
\]
An equality judgement asserts that this type is inhabited.
\end{definition}

\leanref{TSG-DEF-017}{TypedSlottedEGraphsPaper.TypedModel}
A typed model supplies value and environment families, environment
restriction along embeddings, term evaluation, its transport-naturality field,
and local soundness of every primitive equation and structural-congruence rule.
\leanref{TSG-DEF-018}{TypedSlottedEGraphsPaper.TypedModel.evaluate}
The evaluator maps a same-context term and environment to its output value;
the separate naturality field belongs to the preceding model structure.
\leanref{TSG-DEF-019}{TypedSlottedEGraphsPaper.StructuralRealizer.binderBlock}
\leanref{TSG-DEF-020}{TypedSlottedEGraphsPaper.ConcreteStructuralRealizer}
A structural realizer $R$ maps atoms, containers, unary
binders, descriptor blocks, and nodes to source terms; each constructor comes
\leanref{TSG-INV-008}{TypedSlottedEGraphsPaper.StructuralRealizer.binderBlockNatural}
with its one-step naturality equation.  Local laws $Q$ provide only checked
\leanref{TSG-INV-007}{TypedSlottedEGraphsPaper.StructuralLocalLaws.binderBlockAutomorphismCongruence}
atom certificates and forward congruence rules.  In particular, a block
automorphism is admitted through a descriptor-local certificate rule rather
\leanref{TSG-INV-006}{TypedSlottedEGraphsPaper.ClaimLedger.TSGINV006BinderAdmissibilityNaturality}
than by a bare permutation label.  An optional admissibility family may also
record its outer-context restriction law; this specification is separate from
the soundness induction below.

\begin{lemma}[Soundness of Typed $\alpha$-Equivalence Modulo Binder-Block Automorphisms]
\label{lem:alpha-sound}
\leanref{TSG-LEM-007}{TypedSlottedEGraphsPaper.concreteTypedAlphaSemanticSoundness}
\leanref{TSG-ATOM-028}{TypedSlottedEGraphsPaper.paperAtom028AlphaEquivalenceEvaluationTransport}
Given $R$, $Q$, a typed model $\mathfrak M$, and a proof-relevant structural
alpha derivation $q\equiv_{\alpha}^{\rho}r$ under
$\rho\in\TRen(\Gamma,\Gamma')$, every target environment $\eta'$ satisfies
\[
\operatorname{eval}_{\mathfrak M}(\rho\cdot R(q),\eta')=
\operatorname{eval}_{\mathfrak M}(R(r),\eta').
\]
\leanref{TSG-ATOM-029}{TypedSlottedEGraphsPaper.AtomicLedger.binderBlockAutomorphismEvaluationInvariant}
The binder-block case uses exactly the independently supplied
descriptor-local rule in $Q$.
\end{lemma}
\begin{proof}See Appendix~\ref{app:proof-F20}.\end{proof}

The realizer is a function on the structural syntax.  Its constructor-level
naturality fields are the only hypotheses needed for the following equality.

\begin{lemma}[Naturality of Structural Realization]
\label{lem:realization-natural}
\leanref{TSG-LEM-008}{TypedSlottedEGraphsPaper.StructuralRealizer.structuralRealizationNaturality}
For every structural value $n$ and typed embedding $e$,
\[
R(e\cdot n)=e\cdot R(n).
\]
\end{lemma}
\begin{proof}See Appendix~\ref{app:proof-F21}.\end{proof}

\begin{lemma}[Soundness of Direct Node Congruence]
\label{lem:port-sound}
\leanref{TSG-LEM-009}{TypedSlottedEGraphsPaper.concretePortAndNodeCongruenceCertificate}
\leanref{TSG-ATOM-030}{TypedSlottedEGraphsPaper.paperAtom030CertifiedGraphRelativeNodeEquation}
Every proof-relevant direct structural derivation $d$ under a typed embedding
$e$ whose atom steps and constructor rules carry the local certificates in
$Q$ yields
\[
\EqCert_{\mathcal T}(e\cdot R(n),R(n')).
\]
\leanref{TSG-ATOM-031}{TypedSlottedEGraphsPaper.concretePortAndNodeCongruenceSemanticSoundness}
Applying typed-certificate soundness gives the corresponding evaluation
\leanref{TSG-ATOM-032}{TypedSlottedEGraphsPaper.AtomicLedger.sameContextIdentityCongruence}
equality in every typed model; at identity embedding this is the
untransported equation.
\end{lemma}
\begin{proof}See Appendix~\ref{app:proof-F22}.\end{proof}

Let $\mathbf w=(w_a)_{a\in A}$ with
$w_a\in\operatorname{Term}_{S_a}(\tau_a)$.  A raw obligation package gives
only the typed endpoints for stored cases, parent edges, and recorded
symmetry generators; it contains no endpoint certificate.
\begin{definition}[Coherent Equational Certificate Family]
\label{def:eq-cert}
\leanref{TSG-DEF-021}{TypedSlottedEGraphsPaper.PaperF23CoherentCertificateFamily}
The family $\mathbf w$ is coherent with an obligation package $O$ relative to
$\mathcal T_{\Sigma,\mathcal E}$ when it supplies:
\begin{description}
\item[(EC)] for every stored case $c$ of class $a$, with stored context
$T_c$, stored realization $r_c$, and interface inclusion
$\iota_c:S_a\hookrightarrow T_c$, a certificate
\[
\EqCert_{\mathcal T_{\Sigma,\mathcal E}}(T_c;r_c,\iota_c\cdot w_a);
\]
\item[(PC)] for every parent edge $a\xrightarrow{m}b$, a certificate
$\EqCert_{\mathcal T_{\Sigma,\mathcal E}}(S_a;w_a,m\cdot w_b)$; and
\item[(SC)] for every recorded generator $\pi$ at $a$, a certificate
$\EqCert_{\mathcal T_{\Sigma,\mathcal E}}(S_a;w_a,\pi\cdot w_a)$.
\end{description}
\leanref{TSG-ATOM-045}{TypedSlottedEGraphsPaper.paperF23SymmetryWordCertificate}
Identity, inverse, and composition derive (SC) for every represented word in
the recorded generators.
\end{definition}

\leanref{TSG-DEF-022}{TypedSlottedEGraphsPaper.PaperFiniteUnfolding}
A finite parent path contains only raw parent edges.  Its embedding is their
composition.  A certified finite unfolding of $m*a$ to $t$ chooses an actual
stored case $c$, an extension $\bar m:T_c\hookrightarrow\Gamma$ with
$\bar m\circ\iota_c=m$, and a finite trace from $t$ to
$\bar m\cdot r_c$ whose leaves are primitive equations or forward structural
congruence rules (with symmetry permitted).

\begin{lemma}[Find and Finite Unfolding Preserve Equational Certificates]
\label{lem:find-cert}
\leanref{TSG-LEM-010}{TypedSlottedEGraphsPaper.paperF24FindAndFiniteUnfoldingPreserveEquationalCertificates}
\leanref{TSG-ATOM-033}{TypedSlottedEGraphsPaper.paperAtom033FindWitnessCertificate}
Let $p$ be a raw finite parent path from $a$ to a leader $\ell$, let
$\widehat m=m\circ\operatorname{emb}(p)$, and suppose $\mathbf w$ is
coherent.  Then
\begin{equation}
\label{eq:find-witness}
\mathcal T_{\Sigma,\mathcal E}\vdash_\Gamma
m\cdot w_a=\widehat m\cdot w_\ell.
\end{equation}
\leanref{TSG-ATOM-034}{TypedSlottedEGraphsPaper.paperAtom034FiniteUnfoldingWitnessCertificate}
Moreover, every certified finite unfolding of $m*a$ to $t$ yields
\begin{equation}
\label{eq:rep-witness}
\mathcal T_{\Sigma,\mathcal E}\vdash_\Gamma
t=\widehat m\cdot w_\ell.
\end{equation}
\end{lemma}
\begin{proof}See Appendix~\ref{app:proof-F24}.\end{proof}

\begin{lemma}[Certified Leader-Kernel Extraction]
\label{lem:kernel-cert}
\leanref{TSG-LEM-011}{TypedSlottedEGraphsPaper.paperF25CertifiedLeaderKernelExtraction}
Let a certified extraction record contain an original term
$n\in\operatorname{Term}_{\Gamma_0}(\tau)$, an exact kernel
$K\in\operatorname{Term}_{\Delta}(\tau)$, an inclusion
$\iota:\Delta\hookrightarrow\Gamma_0$, and finite provenance whose leaves
are primitive equations or forward structural-congruence rules.  Replaying
that provenance yields
\begin{equation}
\label{eq:kernel-certificate}
d\in
\EqCert_{\mathcal T_{\Sigma,\mathcal E}}
(\Gamma_0;n,\iota\cdot K).
\end{equation}
The theorem validates such a record; constructing it from a concrete graph is
a refinement obligation.
\end{lemma}
\begin{proof}See Appendix~\ref{app:proof-F25}.\end{proof}

\begin{corollary}[Certified Effective-Shape Collision]
\label{cor:collision-cert}
\leanref{TSG-COR-004}{TypedSlottedEGraphsPaper.paperF26CertifiedEffectiveShapeCollision}
Let two certified canonical records have the same output type $\tau$, common
canonical context $C$, shape type $P$, and realization
$\operatorname{real}:P\to\operatorname{Term}_C(\tau)$.  Record $i$ contains
an extracted kernel $K_i$ over $\Delta_i$, a renaming
$\sigma_i:C\cong\Delta_i$, a shape $p_i$, and a certified normalization from
\leanref{TSG-ATOM-037}{TypedSlottedEGraphsPaper.paperAtom037EffectiveKernelCollisionCertificate}
$\sigma_i^{-1}\cdot K_i$ to $\operatorname{real}(p_i)$.  If $p_1=p_2$, put
\[
\rho=\sigma_2\circ\sigma_1^{-1}\in\TRen(\Delta_1,\Delta_2).
\]
Then normalization replay, the shape equality, and inverse replay produce
\[
c_{1,2}\in
\EqCert_{\mathcal T_{\Sigma,\mathcal E}}
(\Delta_2;\rho\cdot K_1,K_2).
\]
\leanref{TSG-ATOM-038}{TypedSlottedEGraphsPaper.paperAtom038AmbientCollisionCertificate}
For any $e_i:\Gamma_i\hookrightarrow\Omega$ satisfying
$e_1\circ\iota_1=e_2\circ\iota_2\circ\rho$, composing $c_{1,2}$ with both
certified extraction records also gives
\[
\EqCert_{\mathcal T_{\Sigma,\mathcal E}}
(\Omega;e_1\cdot n_1,e_2\cdot n_2).
\]
Equality of output type and canonical context is a premise of this typed
collision theorem, not a consequence of untagged payload equality.
\end{corollary}
\begin{proof}See Appendix~\ref{app:proof-F26}.\end{proof}

\begin{theorem}[Certificate-Preserving Reachability and Finite-Unfolding Soundness]
\label{thm:semantic-sound}
\leanref{TSG-THM-001}{TypedSlottedEGraphsPaper.paperF27CertificatePreservingReachabilityAndFiniteUnfoldingSoundness}
\leanref{TSG-ATOM-039}{TypedSlottedEGraphsPaper.paperAtom039CertifiedReachabilityPreservesCoherence}
\leanref{TSG-DEF-023}{TypedSlottedEGraphsPaper.LocalCertifiedTrace}
Let a theorem state expose exactly its live EC, PC, and SC endpoint
obligations.  Suppose a finite proof-relevant trace starts at a state with no
live obligations and ends at the package $O$ of Definition~\ref{def:eq-cert}.
Each step is one of the following nine certified families:
\begin{enumerate}
\leanref{TSG-STEP-001}{TypedSlottedEGraphsPaper.LocalCertifiedStep.insertion}
\leanref{TSG-STEP-008}{TypedSlottedEGraphsPaper.LocalCertifiedStep.rebuild}
\item insertion and rebuilding add obligations with local certificates;
\leanref{TSG-STEP-002}{TypedSlottedEGraphsPaper.LocalCertifiedStep.canonicalization}
\leanref{TSG-STEP-005}{TypedSlottedEGraphsPaper.LocalCertifiedStep.symmetry}
\leanref{TSG-STEP-006}{TypedSlottedEGraphsPaper.LocalCertifiedStep.interfaceRestriction}
\item canonicalization, symmetry update, and interface restriction rekey
obligations using typed endpoint bridges;
\leanref{TSG-STEP-003}{TypedSlottedEGraphsPaper.LocalCertifiedStep.collision}
\leanref{TSG-STEP-004}{TypedSlottedEGraphsPaper.LocalCertifiedStep.union}
\leanref{TSG-STEP-009}{TypedSlottedEGraphsPaper.LocalCertifiedStep.pathCompression}
\item collision, union, and path compression compose two certified
obligations with typed endpoint bridges; and
\leanref{TSG-STEP-007}{TypedSlottedEGraphsPaper.LocalCertifiedStep.noninferentialUnion}
\item noninferential union only removes or inherits obligations.
\end{enumerate}
Then $O$ has a coherent equational certificate family.

\leanref{TSG-ATOM-040}{TypedSlottedEGraphsPaper.paperAtom040FiniteUnfoldingEquationalSoundness}
Moreover, let two raw parent paths end at the same leader, let certified
finite unfoldings in one caller context produce $t$ and $u$, and let a word
in the recorded symmetry generators satisfy the exact equality between the
two composed caller embeddings.  Then
\begin{equation}
\label{eq:derivable-soundness}
\mathcal T_{\Sigma,\mathcal E}\vdash_\Gamma t=u.
\end{equation}
\leanref{TSG-ATOM-041}{TypedSlottedEGraphsPaper.paperAtom041FiniteUnfoldingModelSoundness}
Consequently, for every $\mathfrak M\models\mathcal T_{\Sigma,\mathcal E}$
and $\eta\in\operatorname{Env}_{\mathfrak M}(\Gamma)$,
\begin{equation}
\label{eq:model-soundness}
\operatorname{eval}_{\mathfrak M}(t,\eta)=
\operatorname{eval}_{\mathfrak M}(u,\eta).
\end{equation}
\leanref{TSG-ATOM-042}{TypedSlottedEGraphsPaper.paperAtom042CommonContextWeakenedFiniteUnfoldingSoundness}
The same certificate transports to a common typed context along two equal
embeddings of the caller context.
\end{theorem}
\begin{proof}See Appendix~\ref{app:proof-F27}.\end{proof}

The nine labels above classify abstract obligation transformations, not Java
transition semantics.  The theorem is a conditional validation envelope: it
proves neither that a concrete implementation emits the local evidence nor
unrestricted symmetry, interface restriction, saturation termination, or
semantic completeness.

\paragraph{Formalization status and trust boundary.}
Every retained formal environment carries a distinct correspondence tag.  The
target development is Lean~4~\cite{deMouraUllrich2021Lean4}, pinned to version 4.33.0, under the
\textsc{standard-lean-explicit-tcb} profile, permitting only
\texttt{propext}, \texttt{Quot.sound}, and \texttt{Classical.choice} and
forbidding project axioms, proof placeholders, \texttt{native\_decide},
unsafe/partial proof dependencies, FFI oracles, and unverified solver results.
The complete statement index, proof dependencies, and artifact-refinement
boundary are recorded in Appendix~\ref{app:formal-details} and the 
formal registry at \texttt{paper-claims/formal} in the artifact.  Abstract-kernel closure, Java refinement, and experimental
replay are reported as separate layers; none follows from proof volume or
bounded tests.

\section{Implementation and Evaluation}
\begin{figure}[t]
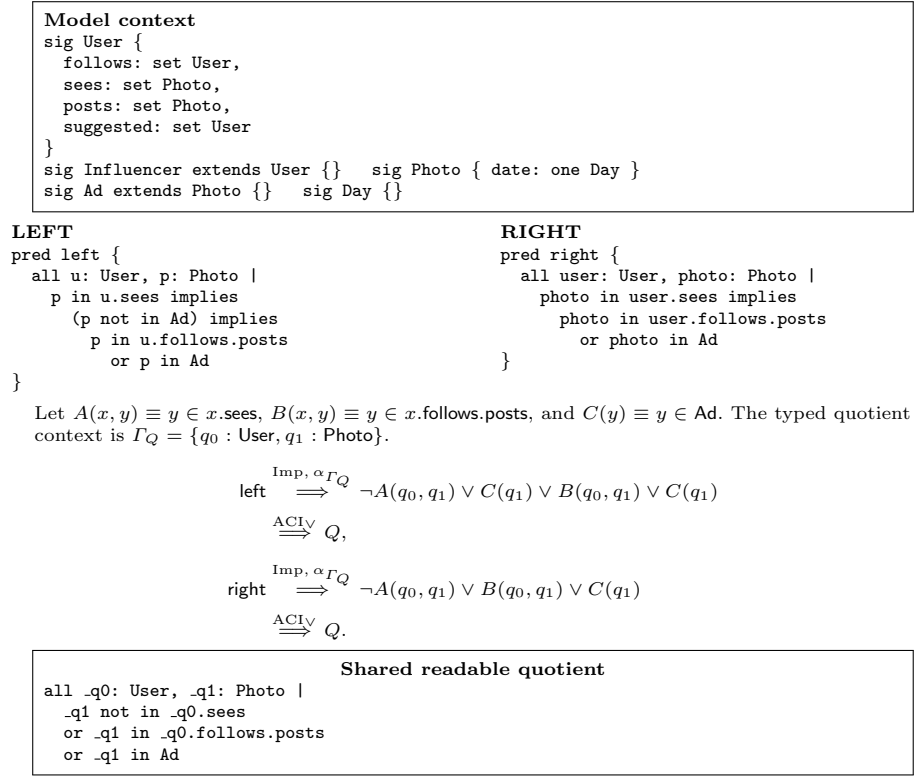

\centering
\scriptsize
\setlength{\fboxsep}{4pt}

\fbox{%
\begin{minipage}{0.93\linewidth}
\textbf{Model context}\par
\ttfamily
sig User \{\\
\hspace*{1em}follows: set User,\\
\hspace*{1em}sees: set Photo,\\
\hspace*{1em}posts: set Photo,\\
\hspace*{1em}suggested: set User\\
\}\\
sig Influencer extends User \{\}\hspace{1.5em}%
sig Photo \{ date: one Day \}\\
sig Ad extends Photo \{\}\hspace{1.5em}sig Day \{\}
\end{minipage}%
}

\vspace{4pt}

\begin{minipage}[t]{0.47\linewidth}
\raggedright
\textbf{LEFT}\par
\ttfamily
pred left \{\\
\hspace*{1em}all u: User, p: Photo |\\
\hspace*{2em}p in u.sees implies\\
\hspace*{3em}(p not in Ad) implies\\
\hspace*{4em}p in u.follows.posts\\
\hspace*{5em}or p in Ad\\
\}
\end{minipage}
\hfill
\begin{minipage}[t]{0.47\linewidth}
\raggedright
\textbf{RIGHT}\par
\ttfamily
pred right \{\\
\hspace*{1em}all user: User, photo: Photo |\\
\hspace*{2em}photo in user.sees implies\\
\hspace*{3em}photo in user.follows.posts\\
\hspace*{4em}or photo in Ad\\
\}
\end{minipage}

\vspace{4pt}

\begin{minipage}{0.95\linewidth}
Let
\(A(x,y)\equiv y\in x.\mathsf{sees}\),
\(B(x,y)\equiv y\in x.\mathsf{follows}.\mathsf{posts}\), and
\(C(y)\equiv y\in\mathsf{Ad}\).
The typed quotient context is
\(\Gamma_Q=\{q_0:\mathsf{User},q_1:\mathsf{Photo}\}\).

\[
\begin{aligned}
\mathsf{left}
&\stackrel{\mathrm{Imp},\,\alpha_{\Gamma_Q}}{\Longrightarrow}
  \neg A(q_0,q_1)\lor C(q_1)\lor B(q_0,q_1)\lor C(q_1)\\
&\stackrel{\mathrm{ACI}_{\lor}}{\Longrightarrow} Q,
\\[1mm]
\mathsf{right}
&\stackrel{\mathrm{Imp},\,\alpha_{\Gamma_Q}}{\Longrightarrow}
  \neg A(q_0,q_1)\lor B(q_0,q_1)\lor C(q_1)\\
&\stackrel{\mathrm{ACI}_{\lor}}{\Longrightarrow} Q.
\end{aligned}
\]
\end{minipage}

\vspace{2pt}

\fbox{%
\begin{minipage}{0.93\linewidth}
{\centering\textbf{Shared readable quotient}\par}
\ttfamily
all \_q0: User, \_q1: Photo |\\
\hspace*{1em}\_q1 not in \_q0.sees\\
\hspace*{1em}or \_q1 in \_q0.follows.posts\\
\hspace*{1em}or \_q1 in Ad
\end{minipage}%
}

\caption{Schematic certificate-bound saturation of two Alloy
predicates. $\mathrm{Imp}$ denotes certified implication elimination;
$\alpha_{\Gamma_Q}$ admits only type-preserving binder renamings, so
the \texttt{User} and \texttt{Photo} coordinates cannot be exchanged;
and $\mathrm{ACI}_{\lor}$ flattens, reorders, and deduplicates the
visible disjunction children under the declared ACI certificates.
The final Alloy text is a readable projection of the common canonical
observation, rather than an independent authority for equality.}
\label{fig:alloy-certificate-example}
\end{figure}
\label{sec:evaluation}
Figure~\ref{fig:alloy-certificate-example} uses the social-media model
introduced in Section~\ref{sec:background} to compare two predicates for
the same policy. Predicate \texttt{left}
uses a nested implication; after implication elimination, its body
contains two copies of \texttt{p in Ad}. Predicate \texttt{right}
states directly that every seen photo is either posted by a followed
user or is an ad. The implication and ACI steps are admitted only by
their endpoint-indexed certificates, while the $\alpha$-equivalence
step maps \texttt{u}/\texttt{user} to \texttt{\_q0: User} and
\texttt{p}/\texttt{photo} to \texttt{\_q1: Photo} through
type-preserving slot maps. The disjunction port then flattens
association, forgets child order, and removes the duplicate under
idempotence. Both inputs therefore have the readable quotient shown
below. This equality is relative to the declared certified rewrite
theory, not to a complete decision procedure for Alloy.

We instantiate the representation with Alloy in an environment with a Java backend consisting of an API based on the Kodkod SAT Solver and a parser. The evaluation
asks (1) how much recorded structural variation the implemented observations
consolidate, (2) which transformation families distinguish the evaluated arms,
and (3) at what recorded cost.  Throughout, \emph{canonical} means canonical
only under the implementation's side-conditioned rewrite vocabulary.  The
frozen measurements are manifest-bound observations: the current evidence does
not establish that the Java adapter refines the formal construction, that the
reported equalities are complete for Alloy semantics, or that the corpus has
been replayed by an independently verified checker. 

\paragraph{Corpus and protocol.}
The classified Alloy4Fun snapshot~\cite{alloy4funDataset2023,alloy4fun,novice} contains 66,080
student--oracle pairs.  Removing 4,482 pairs with identical parsed Abstract Syntax Trees (ASTs) leaves
61,598 pairs: 19,212 labeled \texttt{CORRECT} and 42,386 labeled
overconstrained, underconstrained, or both, in 181 invariant groups and 17
problem variants.  Labels come from the dataset, not from a distance.  The
seven arms are controlled variants of the same Java pipeline, not executions
of external e-graph systems.  Each arm ran once in a fresh process on an AMD
Ryzen~9 9950X3D with 16 workers and an 8~GiB heap cap.  Appendix~\ref{app:evaluation-detail}
records the complete arm definitions, tables, and provenance qualifications.

\paragraph{Observation and equality.}
The pipeline separates an editable repair projection from a deterministic
finite-term equality observation.  On all 61,598 pairs, the mean raw AST has
26.787 nodes and the Certificate-Integrated repair observation has 17.990
units.  The Fast Rewrite Intermediate Representation (IR) normal form averages 18.196 units, and the
Certificate-Integrated materialized canonical-representative tree averages
33.199 nodes; these view-specific sizes are not interchangeable edit budgets.
Within invariant groups, 19,393 recorded truths have 4,496 distinct raw ASTs.
The Certificate-Integrated arm records 2,101 distinct exact observations whose
classes contain 11,382 AST-different equal pairs; Fast Rewrite records 2,136
observations and 10,934 such pairs.  On the paired corpus the
Certificate-Integrated arm assigns zero distance to 4,088 of 19,212
labeled-correct pairs and to none of the 42,386 inherited incorrect pairs.
These are zeroes under the frozen implementation's key, not a refinement proof
or an unbounded semantic soundness result.

\paragraph{Repair distance and coverage.}
For a formula $\varphi$, let $T_\varphi$, $B_\varphi$, and $M_\varphi$ be its
temporal skeleton, binding view, and normalized matrix in the repair
projection.  The reported specification is
\begin{equation}
  d_{\mathrm{repair}}(\varphi,\psi)=
  d_{\mathrm{temp}}(T_\varphi,T_\psi)+
  d_{\mathrm{quant}}(B_\varphi,B_\psi)+
  d_{\mathrm{matrix}}(M_\varphi,M_\psi).
\end{equation}
The coordinates respectively use ordered dynamic programming, exact
type-compatible binder-orbit matching, and carrier-specific sequence,
bag-assignment, set-assignment, or direct-child comparison.
The matrix comparison reuses one binder-owner mapping across every inherited
temporal phase. Proof wrappers,
certificates, rebuild metadata, and support plumbing are absent from the
projection.  The Certificate-Integrated paired mean is 14.022, compared with
13.943 for Fast Rewrite.  Over the 42,386 incorrect predicates, their
respective nearest-truth means are 11.563 and 11.451.

As seen in Figure~\ref{fig:distance-coverage}, at inclusive radii 1, 2, 5, and 10, Certificate-Integrated IR coverage
(blue) is respectively 4.0\%, 10.4\%, 31.0\%, and 59.3\%; raw-AST
coverage (orange) is 8.2\%, 11.7\%, 23.8\%, and 43.0\%.  Each curve uses
its own edit units and independently
selects the nearest AST-distinct recorded truth in the same invariant group;
the vertical difference is therefore not an equal-budget improvement.
Unlike mutation-based specification repair~\cite{mutrepair}, this study does not generate or execute the repair candidates. 

\begin{figure}[t]
\centering
\includegraphics[width={0.78\linewidth}]{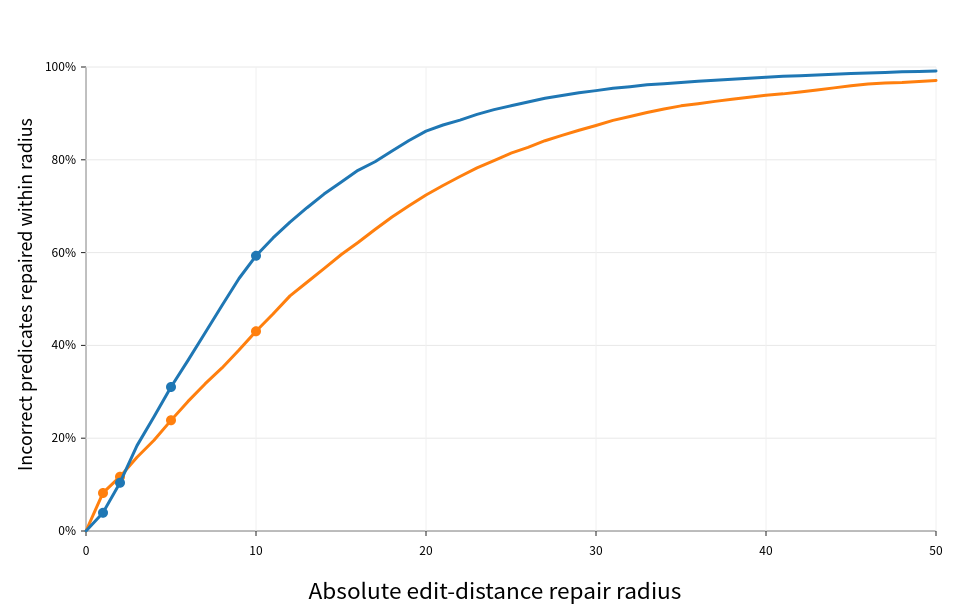}
\caption{Cumulative nearest-reference coverage for the $n=42{,}386$
non-\texttt{CORRECT} predicates.  At inclusive radius $r$, blue is the
Certificate-Integrated IR observation distance and orange is raw-AST tree-edit
distance to the nearest AST-distinct oracle or \texttt{CORRECT} student predicate in
the same invariant group.  The horizontal axis is capped at 50 while all
42,386 predicates remain in the denominator. 
``Repaired'' in the axis label
denotes only falling within this nearest-reference radius.
}
\label{fig:distance-coverage}
\end{figure}

\paragraph{Ablation, capability, and cost.}
Table~\ref{tab:eval-headline} displays the comparisons needed to interpret the
main claims; the full seven-arm ablation and all eleven transformation rows
appear in Appendix~\ref{app:evaluation-detail}.  De Bruijn storage raises the
raw arm's labeled-correct zeroes from 820 to 2,160, whereas the slot-shaped arm
records 2,159.  On the targeted 5,500-pair suite, however, the slot-shaped,
Fast Rewrite, and Certificate-Integrated arms recognize every pair.
Binder-block permutation is the discriminating family: the De Bruijn arm
recognizes 22/500 and each slot-aware arm 500/500.  These are in-vocabulary
recognition results, not completeness or prevalence estimates.

\begin{table}[t]
\centering
\caption{Headline manifest-bound results.  Natural-corpus zeroes and wall time use
61,598 AST-different pairs; targeted recognition uses 5,500 generated,
side-conditioned pairs.  All distances and times are arm-specific.}
\label{tab:eval-headline}
\scriptsize
\setlength{\tabcolsep}{3.3pt}
\begin{tabular}{@{}lrrrr@{}}
\toprule
Arm & Correct $d=0$ & Incorrect $d=0$ & Wall (s) & Targeted\\
\midrule
Raw + De Bruijn & 2,160 & 0 & 19.380 & 3,628/5,500\\
Slot-shaped core & 2,159 & 0 & 19.670 & 5,500/5,500\\
Fast Rewrite IR & 4,074 & 0 & 23.990 & 5,500/5,500\\
Certificate-Integrated & 4,088 & 0 & 2,684.110 & 5,500/5,500\\
\bottomrule
\end{tabular}
\end{table}

The Certificate-Integrated arm behaves as an auditor rather than as a
like-for-like canonicalizer.  For each predicate,  the arm  runs Fast Rewrite IR
and then reconstructs an in-process certificate-checked typed slotted graph.
Each certified mutation validates types and slot contexts,
source-occurrence and binder metadata, and law authority.  Insertions and
rebuild steps invoke whole-state invariant checking, including certificate and
provenance ledgers; congruence merges can re-dirty earlier parent records,
producing a superlinear worst-case work pattern.

The path also enumerates complete bounded finite unfoldings and materializes
proof-heavy construction artifacts.  Several certificate families are
occurrence-bound, while memoization is local to one adaptation, so these
objects cannot be reused wholesale across predicates.  The measured path did
not retain an independent replay trace.  With 16 workers, the certificate-integrated pass  
reaches 7,137.811~MiB peak used heap and 8,943.988~MiB maximum Resident Set Size (RSS).  

Against Fast Rewrite on the same snapshot, this arm took 111.885 times as much
process wall time and used 593.026 times as much engine CPU; on the targeted
suite the process-wall ratio was 86.920. The Certificate-Integrated arm shows the 
average pair-level ablation
representation-unit count slightly smaller, 29.541 versus 30.001.  The observed gap is therefore
consistent with this audit-heavy implementation---global validation, exact
renaming and certificate work, rebuilds, and finite unfolding---rather than
representation growth or a different repair objective.  It is not evidence
of an inherent cost of typed e-graphs or the repair metric, nor an
external-system ranking.

\paragraph{Independent checks and limits.}
A separate verifier reconstructs typed judgments, graph checkpoints,
canonical orbits, and finite unfoldings from closed bundles.  Separately, Lean
checks one fixed nontrivial abstract trace comprising AC reassociation, a
bijective two-slot binder remapping, union insertion, rebuild, and congruence
restoration after operand reordering: its indexed construction enforces exact
adjacency, and replay under its supplied finite-set interpreter proves
endpoint-denotation equality.  This witness establishes neither Java
producer--verifier correspondence nor coverage of other trace shapes.  The
corpus used in-process
checks rather than complete exported-and-replayed traces.  

A manifest-bound
bounded Alloy analysis covered the union of 4,088 current correct-pair
zero-distance claim identities from the Fast Rewrite and
Certificate-Integrated arms. Its supported analyses found no counterexample
or solver error and rejected four deliberately invalid probes. A separate
validation refresh checked one generated example from each of 29
family/subtype combinations, with no counterexample, solver error, or
inconclusive result. Eight checks used bounded temporal solving at object
scope 4 and trace bounds 1--10. This refresh fixes temporal-mode detection
through predicate calls and supersedes the six previously inconclusive
temporal checks; it does not replace the full-corpus measurements.
Newer repository packages also check bounded Java-to-Lean constructor and
certificate-record replays, separately from these corpus observations. 

Overall, full Java--Lean refinement,
full-trace replay, semantic completeness, industrial generalization, and
causal attribution remain future work.


\section{Conclusion}
\label{sec:conclusion}

Typed flexible-arity ports compose certified sibling quotients, separately
licensed flattening, bound-name alpha-renaming, and equality-certified slot
symmetry.  The checked quotient-first theorem is exact for a supplied finite
quotient presentation, and certified records expose the evidence needed for
effective-support extraction and collision.  Abstract local obligation
traces preserve a coherent witness family; two certified unfoldings reaching
one leader and aligned by a recorded symmetry word are equal in every
supplied typed model.  
For our Alloy study, the pipeline recognizes 4,088 AST-distinct correct pairs across 61,598 evaluated pairs with no observed incorrect zero-distance merges and recovered all 5,500 capability cases. 


\label{end:main-text}

\label{start:references}
\bibliographystyle{splncs04}
\bibliography{ref}

\appendix
\clearpage
\section{Formal Statement Index and Proof Details}
\label{app:formal-details}

Table~\ref{tab:formal-index} records the complete named-unit inventory in
dependency order.  The main-text tags map each named unit to its Lean
declaration, while the packaged ledger records the atomic units.

\begin{longtable}{@{}p{0.8cm}p{2.0cm}p{7.7cm}@{}}
\caption{Named formal units.}\label{tab:formal-index}\\
\toprule
ID & Kind & Paper unit\\
\midrule
\endfirsthead
\toprule ID & Kind & Paper unit\\ \midrule
\endhead
F01 & Lemma & Size of a Complete Binary AC Closure\\
F02 & Definition & Types, typed slots, finite contexts, and canonical context\\
F03 & Definition & Typed embeddings, renamings, and permutations\\
F04 & Proposition & Identity, composition, and inverses of typed embeddings\\
F05 & Definition & Port grammar and context-indexed port values\\
F06 & Definition & Typed flexible-arity e-nodes and structural support\\
F07 & Lemma & Typed embeddings preserve port and node types\\
F08 & Corollary & Slot support invariance under embeddings\\
F09 & Corollary & Type safety and functoriality of invocation\\
F10 & Definition & Flexible-arity typed slotted e-graph state and quiescence\\
F11 & Lemma & Leaderization contracts support\\
F12 & Definition & Indexed typed alpha-equivalence of ports\\
F13 & Proposition & Typed alpha-equivalence groupoid laws\\
F14 & Lemma & Transport of graph-relative invocation alignment\\
F15 & Lemma & Identity alignment preserves leader-invocation support\\
F16 & Lemma & Ambient quotient-normal-form exactness\\
F17 & Corollary & Effective-support canonical-shape exactness\\
F18 & Proposition & Finite rebuild quiescence\\
F19 & Definition & Typed equational certificates\\
F20 & Lemma & Semantic soundness of typed alpha-equivalence and block automorphisms\\
F21 & Lemma & Naturality of structural realization\\
F22 & Lemma & Soundness of certified direct node congruence\\
F23 & Definition & Coherent EC/PC/SC certificate family\\
F24 & Lemma & Find and finite unfolding preserve equational certificates\\
F25 & Lemma & Certified leader-kernel extraction\\
F26 & Corollary & Certified effective-shape collision\\
F27 & Theorem & Certificate-preserving reachability and finite-unfolding soundness\\
\bottomrule
\end{longtable}

\subsection{Complete formal statement index}
\label{app:formal-statement-index}
The current census contains 146 active formal units with an exact 146-ID/146-declaration source-tag bijection.  Each row names one distinct Lean declaration; the suffix is the first twelve hexadecimal digits of its exact pretty-type SHA-256.
\begingroup\tiny
\setlength{\tabcolsep}{3pt}
\begin{longtable}{@{}p{1.50cm}p{2.00cm}p{6.85cm}p{1.20cm}@{}}
\toprule ID & Source & Exact Lean declaration & Type hash \\
\midrule\endfirsthead
\toprule ID & Source & Exact Lean declaration & Type hash \\
\midrule\endhead
\texttt{TSG-\allowbreak{}FND-\allowbreak{}001} & \texttt{foundations\_\allowbreak{}compressed.\allowbreak{}tex:11-\allowbreak{}-\allowbreak{}23} & \texttt{Typed\allowbreak{}Slotted\allowbreak{}EGraphs\allowbreak{}Paper.\allowbreak{}ACClosure\allowbreak{}Counts.\allowbreak{}complete\allowbreak{}Binary\allowbreak{}ACClosure\allowbreak{}Counts} & \texttt{795bd01d724e} \\
\texttt{TSG-\allowbreak{}FND-\allowbreak{}004} & \texttt{foundations\_\allowbreak{}compressed.\allowbreak{}tex:26-\allowbreak{}-\allowbreak{}37} & \texttt{Typed\allowbreak{}Slotted\allowbreak{}EGraphs\allowbreak{}Paper.\allowbreak{}Atomic\allowbreak{}Ledger.\allowbreak{}Finite\allowbreak{}Container\allowbreak{}Carrier\allowbreak{}Definition} & \texttt{44bb5b066665} \\
\texttt{TSG-\allowbreak{}DEF-\allowbreak{}001} & \texttt{foundations\_\allowbreak{}compressed.\allowbreak{}tex:39-\allowbreak{}-\allowbreak{}54} & \texttt{Typed\allowbreak{}Slotted\allowbreak{}EGraphs\allowbreak{}Paper.\allowbreak{}Finite\allowbreak{}Typed\allowbreak{}Context.\allowbreak{}f02Canonical\allowbreak{}Context\allowbreak{}Encoding} & \texttt{3f464a97d215} \\
\texttt{TSG-\allowbreak{}DEF-\allowbreak{}002} & \texttt{foundations\_\allowbreak{}compressed.\allowbreak{}tex:56-\allowbreak{}-\allowbreak{}68} & \texttt{Typed\allowbreak{}Slotted\allowbreak{}EGraphs\allowbreak{}Paper.\allowbreak{}paper\allowbreak{}Typed\allowbreak{}Embedding\allowbreak{}Renaming\allowbreak{}Definition} & \texttt{83798f2a89df} \\
\texttt{TSG-\allowbreak{}PROP-\allowbreak{}001} & \texttt{foundations\_\allowbreak{}compressed.\allowbreak{}tex:75-\allowbreak{}-\allowbreak{}89} & \texttt{Typed\allowbreak{}Slotted\allowbreak{}EGraphs\allowbreak{}Paper.\allowbreak{}typed\allowbreak{}Embedding\allowbreak{}Category\allowbreak{}Laws} & \texttt{e4334a5411bf} \\
\bottomrule
\end{longtable}
\endgroup

\subsection{Structural proofs}

\paragraph{F01: complete binary AC closure.}
\label{app:proof-F01}
For distinct generators $X=\{x_1,\ldots,x_n\}$ in the free commutative
semigroup, with no unit or additional equations, every nonempty subset
$Y\subseteq X$ appears as a subproduct in some permutation and
parenthesization, yielding $2^n-1$ product classes.  A binary node is indexed
by an ordered pair of nonempty disjoint subsets $(Y,Z)$.  Giving each
generator one of the states ``left,'' ``right,'' or ``absent'' and excluding
an empty side yields $3^n+1-2^{n+1}$ nodes.  The compiled Lean declaration
constructs bijections to both finite cardinalities from binary and ternary
assignment encodings and proves the two stated natural-number growth
relations.  This is a combinatorial encoding theorem, not a claim about an
operational e-graph saturation procedure.

\paragraph{F04: embedding laws.}
\label{app:proof-F04}
The identity is injective and type preserving; a composite of injective,
type-preserving maps has both properties.  For a proof-relevant typed
renaming, the stored inverse laws give
$\rho^{-1}\circ\rho=\operatorname{id}$ directly.  No finite-cardinality or
surjective-injection argument is used.

\paragraph{F07--F09: typing and support.}
\label{app:proof-F07}
The concrete action is defined recursively on ports, heterogeneous argument
lists, and nodes.  Because schema, bound-scope, and output type are indices,
applying the action immediately supplies inhabitants at the same schema and
output indices; the node-output equality is reflexivity.
\label{app:proof-F08}
The recursive support calculation proves, for every target slot, that support
after action is exactly the direct image of support before action.  Function
extensionality and propositional extensionality turn those pointwise facts
into the two predicate equalities for ports and nodes.  Bound indices are
carried by a separate de Bruijn scope; this theorem introduces no fresh free
coordinate.
\label{app:proof-F09}
F09 unfolds the concrete invocation constructor.  Its result schema contains
the class output by construction, the bare occurrence is definitionally the
identity-embedded invocation, and ambient action definitionally composes the
outer and inner embeddings.

\paragraph{F11: leader support.}
\label{app:proof-F11}
The compiled statement is the range inclusion for two typed embeddings.  If
$s=(m\circ p)(x)$ for caller map $m$ and parent-path embedding $p$, choose
$p(x)$ as the witness that $s\in\operatorname{im}(m)$.  The conclusion is
only invocation-image containment: it neither lifts the result to arbitrary
ports or nodes nor proves independence from a missing coordinate.

\paragraph{F13: alpha groupoid laws.}
\label{app:proof-F13}
The public indexed judgement is the reflexive--symmetric--transitive closure
of the concrete direct structural generators.  Identity, inverse, and
composition are therefore respectively its \texttt{reflexive},
\texttt{symmetric}, and \texttt{transitive} constructors; the composite
renaming is retained in the transitive index.  The unindexed carrier is the
dependent sum of a context and a port value, and its relation existentially
retains the typed renaming, so the same three constructors prove equivalence.
The container and binder clauses belong to the direct generator; F13 does not
independently prove them closed under inversion.

\paragraph{F14: graph-relative transport.}
\label{app:proof-F14}
Unpack the invocation alignment into its two checked parent paths, common
leader, admitted symmetry, and embedding equation.  From
$\rho\circ\widehat m=\widehat m'\circ\pi$ obtain, after independently
renaming the callers,
\[
(\eta'\circ\rho\circ\eta^{-1})\circ\eta\circ\widehat m
=\eta'\circ\widehat m'\circ\pi.
\]
The stored inverse law cancels $\eta^{-1}\circ\eta$.  This is an
invocation-level transport theorem; it asserts no recursive container or
binder transport.

\paragraph{F15: identity-aligned support.}
\label{app:proof-F15}
For leader invocations the identity-alignment equation reduces to
$m=m'\circ\pi$ for an admitted typed permutation $\pi$.  The forward image
direction applies $\pi$ to the leader-slot witness; the reverse direction
applies its stored inverse.  The statement is exactly pointwise equivalence of
the two invocation-image predicates, not a structural-support theorem for
containers or binders.

\paragraph{F16: quotient-normal-form exactness.}
\label{app:proof-F16}
An \texttt{AmbientQuotientPresentation} supplies a complete finite carrier,
fixed local generators, their observation laws, and a linear order.  The
finite-presentation lemmas show that the least enumerated orbit member is
related to its input and that two such least members are equal exactly when
the inputs are in the generated relation.  Generator invariance yields
preservation of support, schema, leaf types, and output; exactness applied to
the normal form yields idempotence.  This proof is generic over the carrier
and does not verify a recursive port/node normalization program.

\paragraph{F17: shape exactness.}
\label{app:proof-F17}
An \texttt{EffectiveSupportCanonicalizer} first extracts each exact kernel and
then applies its structural quotient presentation.  F16's exactness and the
bridge from the unindexed generated relation to its proof-relevant typed
derivation give
\[
\operatorname{Shape}_C(n)=\operatorname{Shape}_C(n')
\quad\Longleftrightarrow\quad
\exists\rho.\ \operatorname{Nonempty}(K_C(n)\Rightarrow_\rho K_C(n')).
\]
The quotient observation law and the two kernel-output equations give output
equality in either direction.  The theorem constructs neither an ambient
input-context renaming nor a least renaming witness.

\paragraph{F18: rebuild quiescence.}
\label{app:proof-F18}
The compiled statement takes an explicit rebuild-step relation, a quiescence
predicate, progress from every nonquiescent state, and a natural-valued strict
decrease witness.  Lean embeds the step relation into strict order on that
natural measure, obtains well-foundedness, and uses well-founded induction.
At a nonquiescent state, progress supplies a decreasing successor; at a
quiescent state, the reflexive trace closes the proof.  No particular graph
worklist or lexicographic implementation measure is part of F18.

\subsection{Certificate and semantic proofs}

\paragraph{F20: semantic alpha-soundness.}
\label{app:proof-F20}
First induct on the concrete alpha-closure derivation.  A direct generator is
converted to an equational certificate by the supplied
\texttt{ConcreteStructuralLaws}; reflexivity uses rename-identity, symmetry
transports the recursively obtained certificate along the inverse renaming,
and transitivity transports and composes certificates using rename
composition.  Typed-certificate semantic soundness then gives the displayed
equality for every target environment.  A descriptor automorphism is usable
only through the policy and local certificate laws supplied to this proof; no
consumer-invariance theorem is inferred from a permutation label.

\paragraph{F21: realization naturality.}
\label{app:proof-F21}
The mutual structural recursor covers complete values and their child
forests.  Each atom, container, unary binder, binder block, and node case
rewrites with the corresponding constructor-level naturality field of the
realizer; the forest cases reduce to list-map composition.  This proves the
syntactic equality $R(e\cdot n)=e\cdot R(n)$ and uses neither fresh-name
choice nor sibling/flat equational laws.

\paragraph{F22: certified congruence.}
\label{app:proof-F22}
The concrete node-alignment structure carries a certified concrete leaf at
every atom step; each leaf contains exactly the required endpoint
certificate.  Recursive port/argument certificate soundness combines those
leaves with the supplied one-constructor laws, and node congruence closes the
outer node.  The model result is obtained by applying typed-certificate
semantic soundness.  Thus F22 assumes the licensed local constructor rules
and per-leaf certificates; it does not manufacture them from graph metadata.

\paragraph{F23: coherent certificate family.}
\label{app:proof-F23}
The raw graph-obligation structure contains only stored-case realizations and
inclusions, parent-edge embeddings, and symmetry-generator renamings.  F23 is
the proposition that every obligation in each of those three endpoint
families has, respectively, an (EC), (PC), or (SC) certificate.  It is a
coherence specification, not an existence theorem; F27 derives an inhabitant
for the final obligation family from a certified trace.

\paragraph{F24: find and finite unfolding.}
\label{app:proof-F24}
For find, induct on the raw finite parent path.  Transport each (PC)
certificate along the preceding embedding, compose it with the recursive
path certificate, and rewrite by rename composition.  A
\texttt{PaperFiniteUnfolding} chooses one actual stored case, a caller
extension compatible with its stored inclusion, and a finite trace whose
local steps are primitive equations or structural-congruence rules (in either
direction).  Induction on that trace gives its endpoint certificate;
transported (EC) and the compatibility equation identify the caller witness.
Composing this result with find proves the two certificates returned by F24.
F24 itself has no redundant-coordinate-independence conclusion.

\paragraph{F25: certified kernel extraction.}
\label{app:proof-F25}
A \texttt{PaperCertifiedKernelExtraction} explicitly records the original
term, exact kernel, its typed inclusion, and a finite
\texttt{CertifiedKernelReplay} from the original to the renamed kernel.  An
induction on that replay maps reflexive, primitive, congruence, symmetric, and
transitive constructors to the corresponding equational-certificate
constructors.  F25 performs this replay only; it does not claim that a
particular extraction algorithm constructs the record.

\paragraph{F26: certified collision.}
\label{app:proof-F26}
Each certified canonical record supplies an exact kernel, a renaming from a
shared canonical context to its support, a shape in one shared shape type,
and normalization replay to one shared function
\texttt{canonicalRealization}.  Shape equality becomes equality of those
canonical terms by congruence.  Transport the left normalization through the
right support renaming, compose with that equality, and reverse the right
normalization to obtain the directed kernel certificate.  Transporting this
certificate and the two F25 certificates to a common context, then rewriting
by the stated embedding-compatibility equation, gives the ambient source
certificate.  F26 assumes shape equality and common output/context indices;
it does not derive them.

\paragraph{F27: trace preservation and finite-unfolding soundness.}
\label{app:proof-F27}
An obligation-free initial state is coherent by elimination.  Induction on
the proof-relevant \texttt{LocalCertifiedTrace} then uses exactly four local
preservation schemas: insertion and rebuild add inherited or locally
certified obligations; canonicalization, symmetry, and interface restriction
rekey an obligation with checked endpoint bridges; collision, union, and path
compression compose two source obligations with checked bridges; and a
noninferential union only removes or inherits obligations.  Converting the
resulting coherence of the exact final endpoint state constructs F23.  These
nine constructors are abstract certificate-transition families, not a Lean
refinement theorem for particular graph code.

For finite-unfolding soundness, F24 sends both unfoldings along their raw
parent paths to one leader.  The recorded (SC) generators close under
identity, inverse, and composition to certify the supplied symmetry word; its
exact embedding equation transports that certificate between the two leader
endpoints.  Transitivity yields the term certificate, typed-model soundness
yields evaluation equality, and transport gives the common-context result
only when the two embeddings of the caller context are equal.

\subsection{Foundation and evidence boundary}

The selected profile is \textsc{standard-lean-explicit-tcb}; the only
permitted foundational axioms are \texttt{propext}, \texttt{Quot.sound}, and
\texttt{Classical.choice}.  Each proof-bearing declaration receives a
per-declaration axiom report.  Java tests, bounded checks, mutation rejection,
and corpus measurements do not establish a Lean refinement.  The final
artifact-refinement and experimental-replay matrices therefore remain
separate from the mathematical statement index.

\subsection{Artifact-refinement boundary}
\label{app:artifact-refinement-boundary}

Table~\ref{tab:artifact-refinement-boundary} renders the archived assessment
at repository snapshot \texttt{6764808}, retained in
\texttt{formal/artifact-refinement-map.json}.  ``Exact Lean specification''
means that the paper-side contract has a distinct compiled declaration; it
does not mean that the listed Java producer refines that contract. No row
in that archived assessment has a Lean refinement proof or a Lean-checked
implementation trace. The newer bounded replay evidence noted in
Appendix~\ref{app:evaluation-detail} is separate from this frozen matrix;
it does not establish the full component refinements listed here.

\begingroup\scriptsize
\setlength{\tabcolsep}{3pt}
\begin{longtable}{@{}p{2.75cm}p{3.0cm}p{3.65cm}p{1.45cm}@{}}
\caption{Artifact-refinement matrix.}\label{tab:artifact-refinement-boundary}\\
\toprule
Bridge row & Paper-side status & Artifact evidence class & Bridge status\\
\midrule\endfirsthead
\toprule
Bridge row & Paper-side status & Artifact evidence class & Bridge status\\
\midrule\endhead
\texttt{ART-TYPED-}\newline\texttt{CARRIERS} & Exact Lean specification & Property/mutation tests & \texttt{OPEN}\\
\texttt{ART-LEADER-}\newline\texttt{KERNEL} & Exact Lean specification & Independent non-Lean checker & \texttt{OPEN}\\
\path{ART-CANONICALIZER} & Exact Lean specification & Property/mutation tests & \texttt{OPEN}\\
\texttt{ART-CERTIFICATE-}\newline\texttt{VERIFIER} & Exact Lean specification & Independent non-Lean checker & \texttt{OPEN}\\
\path{ART-REBUILD} & Exact Lean specification & Property/mutation tests & \texttt{OPEN}\\
\texttt{ART-ALLOY-}\newline\texttt{NORMALIZATION} & Exact Lean rule specifications & Bounded partial repair & \texttt{PARTIAL}\\
\texttt{ART-OPERATOR-LAW-}\newline\texttt{ROUTING} & Exact Lean specification & Independent non-Lean checker & \texttt{OPEN}\\
\texttt{ART-EXPERIMENT-}\newline\texttt{PIPELINE} & Exact scope boundary & Pure Java/run-manifest evidence & \texttt{OPEN}\\
\bottomrule
\end{longtable}
\endgroup

Thus the abstract mathematical layer and the implementation bridge have
different denominators: the former is the 146-unit paper--Lean registry,
whereas the latter consists of these eight Java-facing rows.  Closing the
former cannot upgrade any status in the latter.

\subsection{Edge-case closure summary}
\label{app:edge-case-closure}

Table~\ref{tab:edge-case-closure} renders the layer assessments in
\texttt{formal/edge-case-closure.json}.  ``Proved narrow claim'' records the
exact scoped Lean result, not an end-to-end implementation fact.

\begingroup\scriptsize
\setlength{\tabcolsep}{3pt}
\begin{longtable}{@{}p{3.05cm}p{2.55cm}p{2.15cm}p{3.2cm}@{}}
\caption{Edge-case closure by layer.}\label{tab:edge-case-closure}\\
\toprule
Edge case & Abstract Lean layer & Artifact layer & Scope boundary\\
\midrule\endfirsthead
\toprule
Edge case & Abstract Lean layer & Artifact layer & Scope boundary\\
\midrule\endhead
\texttt{EDGE-PROPER-PARENT-}\newline\texttt{SUPPORT} & Proved narrow claim & \texttt{OPEN} & Parent paths and restriction remain separate obligations.\\
\texttt{EDGE-PARENT-}\newline\texttt{NONINJECTIVITY} & Proved nonimplication & \texttt{OPEN} & Parent equality alone supplies no child symmetry.\\
\texttt{EDGE-QUOTIENT-BEFORE-}\newline\texttt{MINIMUM} & Proved finite-presentation contract & \texttt{OPEN} & The Java orbit producer is not refined.\\
\texttt{EDGE-BINDER-}\newline\texttt{DESCRIPTOR} & Proved certified scope & \texttt{OPEN} & Only descriptor-admitted automorphisms apply.\\
\texttt{EDGE-EMPTY-}\newline\texttt{BINDINGS} & Proved rule specification & \texttt{PARTIAL} & Bounded regressions do not refine the parser/normalizer.\\
\texttt{EDGE-HETEROGENEOUS-}\newline\texttt{JOIN} & Proved operator boundary & Bounded tested & Heterogeneous chains are excluded from homogeneous flattening.\\
\texttt{EDGE-INTEGER-OVERFLOW-}\newline\texttt{PROFILE} & Proved modular-profile scope & \texttt{OPEN} & Adapter/profile correspondence remains unproved.\\
\texttt{EDGE-RAW-SOURCE-}\newline\texttt{REFINEMENT} & Out of scope & \texttt{OPEN} & Raw Alloy parsing and normalization are not Lean-replayed.\\
\texttt{EDGE-BOUNDED-REWRITE-}\newline\texttt{SATURATION} & Out of scope of the finite-rebuild proposition & Bounded tested & A 32-pass rewrite guard is not finite administrative rebuild.\\
\bottomrule
\end{longtable}
\endgroup

\section{Evaluation Detail}
\label{app:evaluation-detail}

This appendix expands the immutable publication measurements summarized in
Section~\ref{sec:evaluation}.  The Certificate-Integrated arm is an in-process
certificate-checked implementation observation; its complete traces were not
exported and independently replayed, and it establishes no Java--theory refinement.

\begingroup
\def\UrlBreaks{\do0\do1\do2\do3\do4\do5\do6\do7\do8\do9%
  \do\a\do\b\do\c\do\d\do\e\do\f}
\Urlmuskip=0mu plus 1mu\relax
The ACGN repository snapshot is
\path{1e2667351532b0c632166fa21ae5fbc7308a8fe7} on branch
\texttt{aislop}. Full-corpus measurements use publication run
\path{db9f89bf-0965-4d74-8080-d9191d5f1aec}, completed on September 8,
2026, from a clean checkout of source commit
\path{8ad5fead39b687d2cadc79b01ac27743c1ece990}.
The natural-corpus dataset SHA-256 is
\path{d6741fbf4c4a9b3714d012d068f84cc918052f1f55211bf4d0443b990736a689},
and the executable JAR SHA-256 is
\path{a053e40e64faa2ecffb0f4999b57aa661eae51933d6daf93a95d973144a71abf}.
The manifest records rule set
\path{canonical-equivalences-v3-explicit-laws} and repair metric
\path{certified-fast-rewrite-repair-distance-v12}.  These identifiers bind the
measurements; they do not certify the implementation against the formal model.
All four experimental stages use that one JAR, 16 workers, an 8~GiB heap,
OpenJDK 17.0.20, and seed 55520260811. The controlled suite has dataset
SHA-256 \path{898d8123ce12ee9a28cb106b801c4d3cb9e1c8aaa2644e0389aedd41e6fb49c3}.
The later assurance-only release retains these full-corpus artifacts.
\endgroup

\subsection{Arms and Representation Boundaries}

\begin{table}[H]
\caption{Evaluation arms and their representation boundaries.}
\label{tab:engine-boundary}
\centering
\scriptsize
\setlength{\tabcolsep}{3pt}
\begin{tabular}{@{}p{2.35cm}p{2.7cm}p{3.05cm}p{3.35cm}@{}}
\toprule
Arm & Stored form & Equality discovery & Boundary\\
\midrule
Raw & ordered fixed-arity keys & bounded concrete rewrite frontier & internal baseline\\
Raw + De Bruijn & fixed keys with DB leaves & same frontier & binding encoding added\\
Java egglog & flattened concrete keys & same frontier & flexible-arity internal baseline\\
Java egglog + De Bruijn & flattened keys with DB leaves & same frontier & combined internal baseline\\
Slotted & ordered invocations and integer slots & frontier plus bounded slot canonicalizer & precursor without the exact contract\\
Fast Rewrite & flexible Alloy repair view & staged normalization and reference metric & high-throughput normalization arm\\
Certificate-Integrated & typed $\Seq/\Bag/\Set/\texttt{One}$ ports & certified graph, orbit, rebuild, finite unfolding & in-process checked observation; no Java--theory refinement\\
\bottomrule
\end{tabular}
\end{table}

The front end constructs a temporal skeleton, moves quantifiers only within a
temporal phase under recorded free-variable and domain conditions, and keeps
non-primitive domain constraints as matrix guards.  It then eliminates
implication, biconditional, and formula-level conditionals, pushes negation,
applies safe prenexing, and performs local relational and logical rewrites.
The certified view stores typed ports, supports, binder descriptors,
witnesses, transition certificates, and bounded finite unfoldings; the repair
projection removes proof and schema plumbing before measuring distance.  The
publication run records the current explicit-law rule set and repair metric.
Its runtime checks and empirical outputs do not prove conformance to the paper
construction.

\subsection{Corpus Views and Distance Summaries}

\begin{table}[H]
\centering
\caption{Mean student-side representation size on the paired corpus
($n=61{,}598$). Counts are view-specific structural units, not bytes,
pair-level ablation units, or directly comparable edit budgets.}
\label{tab:compact}
\small
\begin{tabular}{@{}p{3.1cm}p{5.5cm}r@{}}
\toprule
View & Experimental role & Mean units\\
\midrule
Raw Alloy AST & source-level structural baseline & 26.787\\
Certificate-Integrated repair observation & editable view used by the certified repair metric & 17.990\\
Certificate-Integrated representative tree & deterministic finite tree used for representative TED & 33.199\\
Fast Rewrite normal form & finite term emitted by the Fast Rewrite arm & 18.196\\
\bottomrule
\end{tabular}
\end{table}

\begin{table}[H]
\centering
\caption{Recorded truths within each invariant group. Certificate-Integrated (CI)
classes are equivalence classes of CI equality; equivalent pairs
count AST-different pairs within those classes.}
\label{tab:truth_metrics}
\scriptsize
\setlength{\tabcolsep}{3.4pt}
\begin{tabular}{@{}lrrrr@{}}
\toprule
Problem & Recorded & Unique AST & CI classes & Equivalent pairs\\
\midrule
classroom\_fol & 1,628 & 279 & 88 & 1,303\\
classroom\_rl & 1,133 & 266 & 154 & 252\\
coursesNew & 1,353 & 408 & 196 & 639\\
coursesOld & 2,163 & 527 & 272 & 808\\
cv\_v1 & 110 & 42 & 23 & 29\\
cv\_v2 & 61 & 32 & 25 & 10\\
graphs & 828 & 224 & 116 & 328\\
lts & 256 & 105 & 57 & 110\\
productionLineNew & 703 & 249 & 117 & 530\\
productionLine\_v1 & 149 & 58 & 44 & 17\\
productionLine\_v2 & 1,134 & 354 & 196 & 567\\
socialMedia & 4,558 & 919 & 311 & 4,274\\
trainStationNew & 1,611 & 360 & 171 & 798\\
trainStationOld & 127 & 62 & 54 & 11\\
trash\_fol & 1,677 & 253 & 87 & 1,302\\
trash\_ltl & 883 & 157 & 87 & 126\\
trash\_rl & 1,019 & 201 & 103 & 278\\
\midrule
\textbf{Total} & \textbf{19,393} & \textbf{4,496} & \textbf{2,101} & \textbf{11,382}\\
\bottomrule
\end{tabular}
\end{table}

\begin{table}[t]
\centering
\caption{Distance summaries. Coordinates have different units and are
reported descriptively; a dash means that the augmented nearest-truth run did
not compute that diagnostic.}
\label{tab:distance-summary}
\small
\setlength{\tabcolsep}{4pt}
\begin{tabular}{@{}llcc@{}}
\toprule
Coordinate & Unit & Paired mean & Nearest-truth mean\\
& & $n=61{,}598$ & $n=42{,}386$\\
\midrule
Levenshtein & characters & 39.261 & 28.055\\
Raw tree edit & AST nodes & 22.841 & 15.988\\
Certificate-Integrated repair & projection edits & 14.022 & 11.563\\
Fast Rewrite & projection edits & 13.943 & 11.451\\
Representative TED & finite-tree edits & 32.256 & ---\\
\bottomrule
\end{tabular}
\end{table}

The augmented truth pool contains the oracle and student predicates labeled
\texttt{CORRECT} in the same invariant group, with raw-AST duplicates removed;
comparisons do not cross group boundaries. Each repair coordinate selects its
own nearest reference and has
its own edit unit.  Canonical-representative tree edit distance is a diagnostic
baseline rather than the primary repair metric.

Table~\ref{tab:model-overview} records schema-level context.  Several problem
variants share a base schema, so its ten rows do not correspond one-to-one
with the 17 variants used above.

\begin{table}[H]
\centering
\caption{Base-model size context.  Sig./Abs./Ext./Rel. denote signatures,
abstract signatures, extensions, and relations.}
\label{tab:model-overview}
\scriptsize
\setlength{\tabcolsep}{2.5pt}
\begin{tabular}{@{}lrrrrrrrr@{}}
\toprule
Model & Sig. & Abs. & Ext. & Rel. & Mean arity & Scope & LOC & AST\\
\midrule
classroom & 5 & 0 & 2 & 3 & 2.33 & 4 & 96 & 221\\
courses & 6 & 0 & 2 & 5 & 2.20 & 4 & 88 & 316\\
cv-v1 & 5 & 1 & 2 & 4 & 2.00 & 3 & 36 & 113\\
cv-v2 & 5 & 1 & 2 & 4 & 2.00 & 3 & 36 & 121\\
lts & 3 & 0 & 1 & 1 & 3.00 & 3 & 57 & 165\\
product-line-v1 & 5 & 0 & 2 & 3 & 2.00 & 3 & 36 & 88\\
product-line-v2/v3 & 10 & 1 & 7 & 4 & 2.00 & 3 & 70 & 200\\
social-media & 5 & 0 & 2 & 5 & 2.00 & 3 & 57 & 184\\
train-station-ltl & 7 & 0 & 4 & 3 & 2.00 & 3 & 125 & 570\\
trash & 3 & 0 & 2 & 1 & 2.00 & 3 & 42 & 80\\
\bottomrule
\end{tabular}

\end{table}

\subsection{Full Ablation and Targeted Recognition}

\begin{table}[H]
\centering
\caption{Seven-arm ablation on 61,598 AST-different pairs from the current
manifest-bound run. The final two arms are Fast Rewrite and
Certificate-Integrated; all timings and memory values come from that same run.}
\label{tab:ablation-natural}
\scriptsize
\setlength{\tabcolsep}{3.2pt}
\begin{tabular}{@{}lrrr@{}}
\toprule
Arm & Correct $d=0$ & Incorrect $d=0$ & Mean distance\\
\midrule
Raw & 820 (4.268\%) & 0 & 18.364\\
Raw + De Bruijn & 2,160 (11.243\%) & 0 & 17.911\\
Java egglog & 820 (4.268\%) & 0 & 18.147\\
Java egglog + De Bruijn & 2,160 (11.243\%) & 0 & 17.690\\
Slotted & 2,159 (11.238\%) & 0 & 17.824\\
Fast Rewrite & 4,074 (21.205\%) & 0 & 13.943\\
Certificate-Integrated & 4,088 (21.278\%) & 0 & 14.022\\
\bottomrule
\end{tabular}

\medskip
\setlength{\tabcolsep}{2.8pt}
\begin{tabular}{@{}lrrrr@{}}
\toprule
Arm & Wall (s) & Engine CPU (s) & Max RSS (MiB) & Mean units\\
\midrule
Raw & 19.370 & 4.255 & 1,524.891 & 58.073\\
Raw + De Bruijn & 19.380 & 5.015 & 1,575.672 & 57.752\\
Java egglog & 19.270 & 4.176 & 1,475.656 & 57.055\\
Java egglog + De Bruijn & 18.860 & 4.881 & 1,456.941 & 56.732\\
Slotted & 19.670 & 15.565 & 1,614.004 & 53.168\\
Fast Rewrite & 23.990 & 68.820 & 2,061.391 & 30.001\\
Certificate-Integrated & 2,684.110 & 40,811.847 & 8,943.988 & 29.541\\
\bottomrule
\end{tabular}

\parbox{0.97\linewidth}{\footnotesize\emph{Notes.} Wall is process wall;
engine CPU is the arm's accumulated engine CPU. All arms completed
61,598 pairs without failure; 4,482 raw-AST-identical pairs were skipped.
Percentages use the 19,212 recorded-correct eligible pairs as denominator.
Max RSS is process maximum resident-set size. Distances and pair-level mean
units are arm-specific.}
\end{table}

\begin{table}[H]
\centering
\caption{Targeted transformation-recognition suite. Each row has 500 AST-different pairs
generated by the named, side-conditioned transformation.}
\label{tab:capability-benchmark}
\scriptsize
\setlength{\tabcolsep}{1.55pt}
\begin{tabular}{@{}lrrrrrrr@{}}
\toprule
Transformation & Fix. & F+DB & Var. & V+DB & Slot & Leg. & Exact\\
\midrule
$\alpha$-equivalence & 0 & 500 & 0 & 500 & 500 & 500 & 500\\
ACI & 500 & 500 & 500 & 500 & 500 & 500 & 500\\
Binder-block permutation & 1 & 22 & 1 & 22 & 500 & 500 & 500\\
Safe prenexing & 500 & 500 & 500 & 500 & 500 & 500 & 500\\
Logical normalization & 500 & 500 & 500 & 500 & 500 & 500 & 500\\
Temporal normalization & 500 & 500 & 500 & 500 & 500 & 500 & 500\\
$\alpha$ + AC & 0 & 500 & 0 & 500 & 500 & 500 & 500\\
$\alpha$ + binder permutation & 0 & 22 & 0 & 22 & 500 & 500 & 500\\
Binder permutation + prenex & 22 & 22 & 22 & 22 & 500 & 500 & 500\\
AC + logical normalization & 500 & 500 & 500 & 500 & 500 & 500 & 500\\
Mixed 2--4 transformations & 62 & 62 & 62 & 62 & 500 & 500 & 500\\
\midrule
All 5,500 & 2,585 & 3,628 & 2,585 & 3,628 & 5,500 & 5,500 & 5,500\\
\bottomrule
\end{tabular}

\parbox{0.97\linewidth}{\footnotesize\emph{Scope.} This in-vocabulary
suite tests historical representation and pipeline recognition, not semantic
completeness or prevalence. DB denotes De Bruijn encoding; Var. denotes the
variadic core.}
\end{table}

The compact headers in Table~\ref{tab:capability-benchmark} are preserved
display aliases: Fix., F+DB, Var., V+DB, Slot, Leg., and Exact denote Raw,
Raw + De Bruijn, Java egglog, Java egglog + De Bruijn, Slotted, Fast Rewrite,
and Certificate-Integrated, respectively.

The natural-corpus arms combine different normalization and representation
boundaries, so transitions between them are pipeline-level attributions rather
than estimates for individual rewrite rules.  Absolute timing is compared only
within a run.  Every arm completed the recorded natural-corpus pairs, but the
single fresh-process execution per arm does not support distributional timing
claims or a cross-system performance ranking.

\subsection{Checker Boundary and Threats to Validity}

The separate certificate verifier's adversarial suite covers support
contraction, collision and union, restriction, symmetry, rebuild, path
compression, container laws, binder orbits, and incomplete evidence.  The
standalone Lean module
\path{formal/TypedSlottedEGraphsPaper/TraceEnvelopeWitness.lean} supplies one
fixed nontrivial abstract trace: AC reassociation, a bijective two-slot binder
remapping, union insertion, rebuild, and congruence restoration after operand
reordering.  It constructs the indexed trace with exact adjacency and replays
that trace under an explicit finite-set interpreter to endpoint-denotation
equality, packaging both facts as one checked witness.  This establishes
only inhabitation of that abstract trace shape, not Java producer--verifier
correspondence or coverage of other shapes.  Complete experimental traces were
not exported and replayed.

The bounded Alloy analysis covers the 4,088 claim identities in the current
Certificate-Integrated arm's labeled-correct natural-corpus zero-distance set.  It found no
counterexample or solver error for those identities and rejected four invalid
probes: capture during beta
reduction, comprehension-column permutation, signature shadowing, and a
temporal implication error. Targeted capability validation was refreshed
separately after the checker was repaired to detect temporal operators inside
called predicate bodies. All 29 sampled checks, one per family/subtype
combination, completed without a counterexample, error, or inconclusive
result. Eight use bounded temporal solving at scope 4 with trace bounds
1--10. This includes the six temporal checks that had previously been
inconclusive; the older warned static reduction is superseded by this
bounded temporal evidence.

\begingroup
\def\UrlBreaks{\do0\do1\do2\do3\do4\do5\do6\do7\do8\do9%
  \do\a\do\b\do\c\do\d\do\e\do\f}
\Urlmuskip=0mu plus 1mu\relax
The validation run is \path{659e248c-d3d6-4a2b-8d99-67a0ebcf9eb4},
completed on September 9, 2026, from clean source
\path{8feab00f9190482af6a25334af5b2716653f8ac9}. Its separate JAR has
SHA-256 \path{67e7dd088ef864a9170178e6b6c963a2c339836fa88b25cffe8e20788714f04a}.
The manifest binds the same 5,500 generated models, but the executed solver
sample contains only 29 checks, with one worker and a 1~GiB heap.
The report is \path{capability_validation/soundness.json} within that run.
Neither this validation nor the subsequent bounded constructor and
certificate-record replay packages constitute full-corpus trace replay or
full Java/parser refinement.
\endgroup

The natural corpus consists of short educational specifications with
inherited labels, and excluding identical ASTs conditions the analysis on
syntactic difference.  Distances are not calibrated semantic-error
probabilities.  Natural and targeted artifacts have separate run and dataset
identities, so comparisons remain within their respective manifests.  These
data do not establish Java--theory refinement, full semantic equivalence,
industrial generalization, or operational repair success.


\section{Canonicalization Algorithm and Rewrite Theory}
\label{app:canon-rules}

This appendix makes explicit the proposed canonicalizer and its abstract
rewrite contract.  The Lean development proves finite-orbit normalization and
the Boolean, list, finite-function, and relation equations identified below.
Figure~\ref{alg:canon-g} and the Alloy/Java names are a refinement target for
those abstract objects, not a proved implementation correspondence.  An arrow
below denotes a proposed representative orientation unless the text marks the
operation as structural cleanup.  Publication run \path{db9f89bf-0965-4d74-8080-d9191d5f1aec}
records rule set \path{canonical-equivalences-v3-explicit-laws}.  This binds the
measured configuration but does not prove that Java realizes this appendix's
abstract rules or contracts.

\subsection{Notation and Preconditions}

Let $\mathcal G=(U,M,H)$ be a flexible-arity typed slotted e-graph.  The parent
assignment $U$ maps each e-class identifier to an immediate parent invocation,
with a root mapped to its identity invocation;
$\operatorname{find}_\Gamma$ composes these links to return a leader invocation
in caller context $\Gamma$.  The map
$M(a)=(\tau_a,S_a,B_a,G_a)$ stores the result type, exposed typed slots,
canonical shapes, and justified slot-symmetry group of e-class $a$, and $H$ is
the shape collision index.  A typed e-node has the form
\[
  n=f_\theta(q_1,\ldots,q_r),
  \qquad
  \Sigma(f)=\forall\bar\alpha.(\kappa_1,\ldots,\kappa_r)\Rightarrow\tau,
\]
where $q_i\in\operatorname{Port}_\Gamma(\theta(\kappa_i))$.  The schemas
$\texttt{Seq}$, $\texttt{Bag}$, and $\texttt{Set}$ respectively preserve order
and multiplicity, quotient by permutation while preserving multiplicity, and
quotient by both permutation and idempotence.  A first-class
$\texttt{BindBlock}(\beta,\kappa)$ carries a complete descriptor $\beta$, a
fresh occurrence-local bound context, and precisely the certified finite group
$\operatorname{Aut}(\beta)$.  Fix a total structural order $\triangleleft$ over
complete well-typed port and node keys, and a total order $\prec$ over completed
typed shapes.  The latter compares the operator, type instantiation, atom, slot
interface, and normalized ports.  Fix also the same total map order as in
the main text,
$\triangleleft_{\mathrm{map}}:=\prec_{\mathrm{ren}}$, for deterministic witness selection, and
let $\triangleleft_{\mathrm{lex}}$ compare a shape--map pair first by
$\prec$ and then by $\triangleleft_{\mathrm{map}}$.

At the specification level, leader normalization may reduce support, so the
proposed graph operation distinguishes the input context from the effective
kernel context.  For a node $n$ with
$\Gamma_0=\operatorname{slots}(n)$, its intended structural extraction is
\begin{align*}
\operatorname{leaderKernelTrace}_{\mathcal G}(n)
  &=(K_{\mathcal G}(n),\iota_n,\xi_n),\\
\Delta_n&=\operatorname{slots}(K_{\mathcal G}(n))\subseteq\Gamma_0,
&\iota_n&:\Delta_n\hookrightarrow\Gamma_0.
\end{align*}
Here $K_{\mathcal G}(n)$ and $\xi_n$ denote the proposed kernel and trace.
\leanref{TSG-APP-001}{TypedSlottedEGraphsPaper.ProfileRules.tsgApp001_kernelTraceReplayContract}
The exact formal claim \textsc{app-001} uses an indexed trace interface.  For an
indexed family $\mathsf{Term}:\mathbb N\to\mathsf{Type}$, an extractor returns
an effective size $k$, a kernel in $\mathsf{Term}(k)$, and a structural trace
with dependent endpoints
$\langle n,\mathit{input}\rangle$ and $\langle k,\mathit{kernel}\rangle$.
For any two trace interpreters, each supplying an equality for every trace
label, replay proves equality of the corresponding endpoint
denotations.  This theorem neither identifies those equalities with the
paper's $\operatorname{EqCert}$ type nor proves that Java extraction emits the
required trace; those are separate refinement obligations.

For a finite typed context $\Delta$, $\operatorname{Can}(\Delta)$ is the
context containing the same number of
slots of each type, named from a fixed type-indexed alphabet.  The shape
witness for $n$ is
\[
  \sigma_n\in\operatorname{TRen}(\operatorname{Can}(\Delta_n),\Delta_n),
\]
which maps canonical slots bijectively to the effective source support.  Its
transport into the original ambient context is the typed embedding
\[
  \omega_n=\iota_n\circ\sigma_n
  \in\operatorname{TEmb}(\operatorname{Can}(\Delta_n),\Gamma_0).
\]
The map $\omega_n$ is classified as a typed embedding; it is a renaming in the
support-preserving case $\Delta_n=\Gamma_0$.  This separation retains the
shape's exact slots, its effective support, and its original ambient support,
including coordinates removed by find.  Binder-block automorphisms act on
exactly the bound coordinates admitted by the complete
descriptor: domain, quantifier, cardinality/multiplicity, disjointness class,
dependency constraints, and type.  Each block occurrence is alpha-aligned with
a fresh canonical variant of the descriptor's bound context; the descriptor
group is transported to that occurrence context before it acts.  This local
action is disjoint from the one global free-slot renaming.

As the smallest strict-support case, suppose
$U(a)=m*\ell$ with a proper embedding $m:S_\ell\hookrightarrow S_a$.
Then $\operatorname{wrap}(\operatorname{id}_{S_a}*a)$ has input support $S_a$,
whereas its leader kernel $\operatorname{wrap}(m*\ell)$ has support
$m[S_\ell]$.  The construction below therefore uses
$\operatorname{Can}(m[S_\ell])$ and keeps
the inclusion into $S_a$ separately.  Its shape witness is the bijection onto
the effective support, while its ambient transport is the proper embedding
$\operatorname{Can}(m[S_\ell])\hookrightarrow S_a$.

\subsection{The Graph-Relative Canonicalizer}
\label{app:canon-g}

\leanref{TSG-APP-002}{TypedSlottedEGraphsPaper.ProfileRules.tsgApp002_completeCanonicalRecords}
Figure~\ref{alg:canon-g} gives the proposed graph-level pseudocode.  The exact
formal claim \textsc{app-002} concerns an abstract structural record extending
the extractor result with a shape and an effective-size renaming.  Its ambient
transport is definitionally the inclusion composed with that renaming, and its
collision key is definitionally the shape.  For any trace interpreter, replay
adds only equality of the two endpoint denotations.  No completeness of a
local orbit, graph insertion behavior, or $\operatorname{EqCert}$ construction
is part of \textsc{app-002}.  The paper tuple
\[
  (K_{\mathcal G}(n),\operatorname{Shape}_{\mathcal G}(n),\sigma_n,
   \iota_n,\omega_n,\xi_n)
\]
is the intended concrete realization of that record; establishing that the
Java record realizes it remains an artifact-refinement task.

\leanref{TSG-APP-003}{TypedSlottedEGraphsPaper.ProfileRules.tsgApp003_canonicalBlockOccurrenceContract}
The exact formal block-occurrence claim \textsc{app-003} uses de Bruijn
indices of one fixed finite size.  The canonical occurrence is the identity,
the alignment is the inverse of the source occurrence, and the returned list
is the descriptor's supplied automorphism list mapped by conjugation with that
identity.  Consequently source occurrence followed by alignment is the
identity occurrence.  Fresh-name choice, disjointness from a caller, leastness,
and completeness of any richer paper descriptor are not conclusions of this
claim; the corresponding steps in Figure~\ref{alg:canon-g} are refinement
requirements.

\providecommand{\algline}[2][0pt]{%
  \hangindent=#1\hangafter=1\hspace*{#1}#2\strut\\%
}

\begingroup
\small
\setlength{\LTleft}{0pt}
\setlength{\LTright}{0pt}
\setlength{\LTpre}{0pt}
\setlength{\LTpost}{0pt}

\Needspace{18\baselineskip}
\noindent\captionof{figure}{Proposed graph-relative canonicalization refinement
target (not verified as a producer by Lean).  The design minimizes invocation
and binder-block orbits before bag aggregation and set deduplication.  It would
retain $K_{\mathcal G}(n)$, $\iota_n$, $\omega_n$, and $\xi_n$ in the structural
record, obtain a dependent certificate $d_n^{\mathbf w}$ only by separate
replay, and use only $p_{\min}$ as a hash key.  The exact checked contracts are
stated after the figure.}
\label{alg:canon-g}

\begin{longtable}{@{}>{\raggedright\arraybackslash}p{0.98\linewidth}@{}}
\hline
\endfirsthead

\multicolumn{1}{@{}l@{}}{\small\itshape Figure~\thefigure\ continued}\\
\hline
\endhead

\hline
\endfoot

\endlastfoot

\algline{\textbf{status:} specification-only pseudocode; every producer and
  Java correspondence below is a separate refinement obligation}

\noalign{\vskip 0.35em}

\algline{\textbf{procedure} $\operatorname{leaderKernelTrace}_{\mathcal G}(n)$}
\algline[2.4em]{\textbf{require} $n$ is well typed and
  $\Gamma_0=\operatorname{slots}(n)$ is its exact context}
\algline[2.4em]{$(\widehat n_{\mathcal G},\xi)\leftarrow
  \operatorname{findAllWithTrace}_{\mathcal G}(n)$
  \hfill\textit{// find every invocation; do not unfold leaders}}
\algline[2.4em]{$(\widehat n_{\mathcal G},\xi)\leftarrow
  \operatorname{normalizeSignatureAdmittedPortsWithTrace}
  (\widehat n_{\mathcal G},\xi)$}
\algline[2.4em]{$\Delta\leftarrow\operatorname{slots}(\widehat n_{\mathcal G});
  \iota\leftarrow\operatorname{inclusion}(\Delta,\Gamma_0)$}
\algline[2.4em]{$n^\downarrow\leftarrow
  \operatorname{restrictExactContext}(\widehat n_{\mathcal G},\Delta)$}
\algline[2.4em]{\textbf{return} $(n^\downarrow,\iota,\xi)$}

\noalign{\vskip 0.35em}

\algline{\textbf{procedure}
  $\operatorname{canon}_{\mathcal G}(n=f_\theta(q_1,\ldots,q_r))$}
\algline[2.4em]{\textbf{require} $n$ is well typed and already uses its
  declared sibling quotients and explicit flat licenses}
\algline[2.4em]{$(K,\iota_n,\xi_n)\leftarrow
  \operatorname{leaderKernelTrace}_{\mathcal G}(n)$}
\algline[2.4em]{write $K=f_\theta(\bar q_1,\ldots,\bar q_r)$}
\algline[2.4em]{$\Delta\leftarrow\operatorname{slots}(K)$;
  $C\leftarrow\operatorname{Can}(\Delta)$}
\algline[2.4em]{$b\leftarrow\mathsf{None}$
  \hfill\textit{// option of a shape--renaming pair}}
\algline[2.4em]{\textbf{for each} $\sigma\in\operatorname{TRen}(C,\Delta)$ \textbf{do}}
\algline[4.8em]{$\rho\leftarrow\sigma^{-1}$
  \hfill\textit{// effective source slots to canonical slots}}
\algline[4.8em]{\textbf{for} $i\leftarrow1$ \textbf{to} $r$ \textbf{do}}
\algline[7.2em]{$r_i\leftarrow\operatorname{canonLeaderPort}_{\mathcal G}
  (\theta(\kappa_i),\bar q_i,\rho)$}
\algline[4.8em]{$p\leftarrow f_\theta(r_1,\ldots,r_r)$}
\algline[4.8em]{\textbf{if} $b=\mathsf{None}$ or
  $(p,\sigma)\mathrel{\triangleleft_{\mathrm{lex}}}\operatorname{get}(b)$
  \textbf{then} $b\leftarrow\mathsf{Some}(p,\sigma)$}
\algline[2.4em]{\textbf{require} $b=\mathsf{Some}(p_{\min},\sigma_{\min})$
  \hfill\textit{// $\operatorname{TRen}(C,\Delta)$ contains a bijection}}
\algline[2.4em]{$\omega_{\min}\leftarrow\iota_n\circ\sigma_{\min}$}
\algline[2.4em]{\textbf{return}
  $(K,p_{\min},\sigma_{\min},\iota_n,\omega_{\min},\xi_n)$}

\noalign{\vskip 0.35em}

\algline{\textbf{procedure} $\operatorname{replayKernelCertificate}_{\mathcal G,\mathbf w}
  (n,K,\iota,\xi)$}
\algline[2.4em]{\textbf{require} $\mathbf w$ is a coherent witness family for $\mathcal G$}
\algline[2.4em]{\textbf{require} $(K,\iota,\xi)=\operatorname{leaderKernelTrace}_{\mathcal G}(n)$
  and every certificate named by $\xi$ verifies}
\algline[2.4em]{$\Gamma_0\leftarrow\operatorname{slots}(n)$}
\algline[2.4em]{$d^{\mathbf w}\leftarrow
  \operatorname{replayTrace}_{\mathcal G,\mathbf w}(n,K,\iota,\xi)$}
\algline[2.4em]{\textbf{require}
  $d^{\mathbf w}:\operatorname{EqCert}_{\mathcal T_{\Sigma,\mathcal E}}
  (\Gamma_0;\lfloor n\rfloor_{\mathbf w},
  \iota\mathbin\cdot\lfloor K\rfloor_{\mathbf w})$}
\algline[2.4em]{\textbf{return} $d^{\mathbf w}$}

\noalign{\vskip 0.35em}

\algline{\textbf{procedure} $\operatorname{canon}_{\mathcal G}^{\mathbf w}(n)$}
\algline[2.4em]{$(K,p,\sigma,\iota,\omega,\xi)\leftarrow\operatorname{canon}_{\mathcal G}(n)$}
\algline[2.4em]{$d^{\mathbf w}\leftarrow
  \operatorname{replayKernelCertificate}_{\mathcal G,\mathbf w}
  (n,K,\iota,\xi)$}
\algline[2.4em]{\textbf{return}
  $(K,p,\sigma,\iota,\omega,\xi,d^{\mathbf w})$}

\noalign{\vskip 0.35em}

\algline{\textbf{procedure}
  $\operatorname{canonLeaderPort}_{\mathcal G}(\kappa,q,\rho)$}
\algline[2.4em]{\textbf{require}
  $q\in\operatorname{LPort}_{\mathcal G,\Lambda}(\kappa)$ and
  $\rho\in\operatorname{TRen}(\Lambda,\Lambda')$}
\algline[2.4em]{\textbf{target postcondition} result is
  $Q_{\mathcal G}^{\Lambda'}(\rho\mathbin\cdot q)$}
\algline[2.4em]{\textbf{match} $(\kappa,q)$ \textbf{with}}
\algline[4.8em]{$(\texttt{One}(\tau),\mathord{\$}x)$:\quad
  \textbf{return} $\rho(\mathord{\$}x)$}
\algline[4.8em]{$(\texttt{One}(\tau),m*\ell)$:}
\algline[7.2em]{\textbf{require} $\ell$ is a leader;
  $b\leftarrow\mathsf{None}$}
\algline[7.2em]{\textbf{for each} $\pi\in G_\ell$ \textbf{do}}
\algline[9.6em]{$u\leftarrow((\rho\circ m)\circ\pi)*\ell$}
\algline[9.6em]{$b\leftarrow\operatorname{keepLeast}_{\triangleleft}(b,u)$}
\algline[7.2em]{\textbf{return} $\operatorname{get}(b)$
  \hfill\textit{// $\operatorname{id}\in G_\ell$}}
\algline[4.8em]{$(\texttt{Seq}^{J}(\kappa'),[q_1,\ldots,q_k])$:}
\algline[7.2em]{\textbf{return}
  $[\operatorname{canonLeaderPort}_{\mathcal G}
    (\kappa',q_i,\rho)]_{i=1}^{k}$}
\algline[4.8em]{$(\texttt{Bag}^{J}(\kappa'),b)$:}
\algline[7.2em]{$b'\leftarrow$ the empty finite multiplicity map}
\algline[7.2em]{\textbf{for each} $q'\in\operatorname{supp}(b)$
  \textbf{do}}
\algline[9.6em]{$u\leftarrow\operatorname{canonLeaderPort}_{\mathcal G}
  (\kappa',q',\rho)$;
  $b'(u)\leftarrow b'(u)+b(q')$}
\algline[7.2em]{\textbf{return} $b'$
  \hfill\textit{// serialization order is not part of the bag}}
\algline[4.8em]{$(\texttt{Set}^{J}(\kappa'),Y)$:}
\algline[7.2em]{\textbf{return}
  $\{\operatorname{canonLeaderPort}_{\mathcal G}
    (\kappa',q',\rho)\mid q'\in Y\}$}
\algline[4.8em]{$(\texttt{Bind}(\tau,\kappa'),
  \operatorname{bind}\mathord{\$}x{:}\tau.q)$:}
\algline[7.2em]{$\mathord{\$}y\leftarrow
  \operatorname{freshCanonicalSlot}(\tau,\operatorname{cod}(\rho))$}
\algline[7.2em]{$\rho'\leftarrow
  \rho\mathbin{\uplus}[\mathord{\$}x\mapsto\mathord{\$}y]$}
\algline[7.2em]{\textbf{return}
  $\operatorname{bind}\mathord{\$}y{:}\tau.
  \operatorname{canonLeaderPort}_{\mathcal G}(\kappa',q,\rho')$}
\algline[4.8em]{$(\texttt{BindBlock}(\beta,\kappa'),
  \operatorname{bind}_{\beta}^{o}q)$:}
\algline[7.2em]{$(o_B,\alpha_o,A_B)\leftarrow
  \operatorname{canonicalBlockOccurrence}(\beta,o,
  \operatorname{cod}(\rho))$}
\algline[7.2em]{\hspace{1em}\textit{// $o_B=(B,a_B)$ is fresh;
  $\alpha_o=a_B\circ a_o^{-1}$; $A_B$ is the transported group}}
\algline[7.2em]{$b\leftarrow\mathsf{None}$}
\algline[7.2em]{\textbf{for each} $\pi_B\in A_B$ \textbf{do}}
\algline[9.6em]{$\rho_{\pi}\leftarrow
  \rho\mathbin{\uplus}(\pi_B\circ\alpha_o)$}
\algline[9.6em]{$u\leftarrow\operatorname{bind}_{\beta}^{o_B}
  \operatorname{canonLeaderPort}_{\mathcal G}(\kappa',q,\rho_{\pi})$}
\algline[9.6em]{$b\leftarrow\operatorname{keepLeast}_{\triangleleft}(b,u)$}
\algline[7.2em]{\textbf{return} $\operatorname{get}(b)$
  \hfill\textit{// $\operatorname{id}\in\operatorname{Aut}(\beta)$}}
\hline
\end{longtable}
\endgroup

\paragraph{Exact formal contracts for the figure.}
The algorithm rows are read through the following abstract claims, not as a
verified execution of the displayed graph program.
\leanref{TSG-ALG-005}{TypedSlottedEGraphsPaper.ProfileRules.tsgAlg005_leaderKernelTracePseudocodeContract}
\textsc{alg-005} says that an abstract kernel extractor's returned finite
inclusion is injective and that its trace replays to endpoint-denotation
equality under a supplied interpreter.
\leanref{TSG-ALG-006}{TypedSlottedEGraphsPaper.ClaimLedger.TSGALG006CompleteCanonPseudocodeContract}
\textsc{alg-006} uses an extensionally complete nonempty finite candidate list:
every enumerated candidate is admissible, every admissible candidate is
enumerated, the selected candidate is present and least, the returned record
contains exactly it, and the collision key projects only the returned shape.
These are finite-presentation contracts, not a proof that the concrete
$Q_{\mathcal G}$ enumerator has them.
\leanref{TSG-ALG-007}{TypedSlottedEGraphsPaper.ProfileRules.tsgAlg007_replayKernelCertificatePseudocodeContract}
\textsc{alg-007} says that a structural trace between two dependent sigma
endpoints replays to equality under a supplied interpreter.
\leanref{TSG-ALG-008}{TypedSlottedEGraphsPaper.ProfileRules.tsgAlg008_certifiedCanonWrapperPseudocodeContract}
\textsc{alg-008} says that the certified wrapper preserves its input
structural record and adds that replayed equality.
\leanref{TSG-ALG-009}{TypedSlottedEGraphsPaper.ProfileRules.tsgAlg009_canonLeaderPortPseudocodeContract}
Finally, \textsc{alg-009} proves constructor-by-constructor that the abstract
recursive function realizes its independently defined \emph{untyped}
$\operatorname{PortQuotient}$ relation: slots are mapped, invocation candidates
are passed to the defined $\operatorname{leastByKey}$ function, sequence and bag lists are recursively
mapped, set lists are deduplicated by key, unary binders recurse, and a block
applies the same function to its supplied candidates.  It does not prove typed-schema
admissibility, graph lookup, or completeness of the supplied invocation or
block candidates.

\leanref{TSG-APP-004}{TypedSlottedEGraphsPaper.ProfileRules.tsgApp004_canonLeaderPortContract}
The exact finite-orbit claim \textsc{app-004} is correspondingly generic.  For
any finite normalization presentation $P$, $P.\mathit{normal}(x)$ belongs to
the supplied orbit of $x$ and is $P$-related to $x$, and
\[
  P.\mathit{normal}(x)=P.\mathit{normal}(y)
  \quad\Longleftrightarrow\quad P.\mathit{relation}(x,y).
\]
This is the finite-presentation normal-form theorem used by the paper; identifying $P$ with the
concrete graph-relative relation is an additional refinement obligation.

\leanref{TSG-APP-005}{TypedSlottedEGraphsPaper.ClaimLedger.TSGAPP005CertifiedAbstractRebuildContract}
The exact rebuild claim \textsc{app-005} is an abstract three-part contract.
For a supplied finite-decreasing rebuild system, every initial state reaches
some quiescent state.  A supplied interface-restriction certificate decomposes
the old natural-number size as the retained size plus further removals plus
one, so committing it strictly decreases the size.  Finally, pointwise port
evidence at one shared head and type instantiation constructs the corresponding
node-congruence witness.  Concrete collision buckets, transport of graph
witnesses, pending-certificate behavior, and Java fixed-point scheduling are
not conclusions of this theorem; they remain implementation obligations.

\subsection{Ordered Application of Rules}
\label{app:rule-order}

\leanref{TSG-RULE-001}{TypedSlottedEGraphsPaper.ProfileRules.tsgRule001_orderedRulePipeline}
The formal object for \textsc{rule-001} is the displayed list of nine stage
names: (1) parser cleanup; (2) alpha normalization; (3) beta reduction;
(4) branch elimination; (5) first NNF; (6) phase-local prenexing; (7) guard
elimination and NNF; (8) sibling and flat normalization; and (9) local
saturation and rebuild.  Its only phase theorem is that the formal
$\operatorname{normalizeWithinPhase}$ operation, presently the identity,
preserves the stored phase index.  It does not establish that the Java stages
execute in this order, that temporal children are normalized recursively, or
that arbitrary quantifier motion cannot cross a phase.  Those are artifact
conformance requirements.

\subsection{Parser and Structural Normalization}

\leanref{TSG-RULE-002}{TypedSlottedEGraphsPaper.ProfileRules.tsgRule002_parserHeadNormalization}
For \textsc{rule-002}, each finite parser-head constructor in
Table~\ref{tab:opcode-aliases} maps to the shown canonical-head constructor;
the enumerated operation-meaning label of the result equals the enumerated
label of the source.  This is a datatype equation, not a denotational theorem
about the Java parser or Alloy expressions.

\begin{table}[t]
\centering
\small
\caption{Parser heads mapped to canonical internal heads.}
\label{tab:opcode-aliases}
\setlength{\tabcolsep}{4pt}
\begin{tabular}{ll@{\qquad}ll}
\hline
Source heads & Canonical head & Source heads & Canonical head\\
\hline
\texttt{BF/AND}, \texttt{LF/AND} & \texttt{BOOL/AND} &
\texttt{BF/OR}, \texttt{LF/OR} & \texttt{BOOL/OR}\\
\texttt{BF/IMPLIES} & \texttt{BOOL/IMPLIES} &
\texttt{BF/IFF} & \texttt{BOOL/IFF}\\
\texttt{UF/NOT} & \texttt{BOOL/NOT} &
\texttt{BE/PLUS} & \texttt{REL/PLUS}\\
\texttt{BE/INTERSECT} & \texttt{REL/INTERSECT} &
\texttt{BE/JOIN} & \texttt{REL/JOIN}\\
\texttt{BE/ARROW} & \texttt{REL/ARROW} &
\texttt{BE/MUL} & \texttt{ARITH/MUL}\\
\texttt{BE/IPLUS} & \texttt{ARITH/PLUS} &
\texttt{LE/DISJOINT} & \texttt{LIST/DISJOINT}\\
\hline
\end{tabular}
\end{table}

\leanref{TSG-RULE-003}{TypedSlottedEGraphsPaper.ProfileRules.tsgRule003_parserStructuralCleanup}
Internal no-op and end markers are removed:
\[
  \operatorname{NOOP}(E)\longrightarrow E,
  \qquad
  \langle\mathrm{END}\rangle\longrightarrow\epsilon.
\]
The exact \textsc{rule-003} equations say that a singleton \textsc{noop} list
unwraps, a singleton \textsc{end} list deletes, and a singleton abstract
\textsc{incomplete} list deletes.  They contain no arity validator and no claim
about Java parser rejection.  Lexical alpha-renaming and diagnostic-name
retention are likewise implementation requirements, not conclusions of this
rule.

\leanref{TSG-RULE-004}{TypedSlottedEGraphsPaper.ProfileRules.tsgRule004_captureAvoidingLetBeta}
For intrinsically scoped de Bruijn terms, \textsc{rule-004} defines
\[
  \operatorname{let}\ x=E\mid P\longrightarrow P[E/x].
\]
Precisely, reducing the formal $\mathsf{letE}(E,P)$ returns
$\mathsf{some}(\operatorname{substituteTop}(E,P))$, and lifted substitution
fixes the newly bound coordinate $0$.  The index types enforce scoping; the
theorem does not assert semantic preservation for a concrete Alloy AST or its
Java substitution routine.

\subsection{Branch Connectives and Boolean NNF}

\leanref{TSG-RULE-005}{TypedSlottedEGraphsPaper.ProfileRules.tsgRule005_branchConnectiveElimination}
The third rule below applies to formula-valued conditionals:
\begin{align*}
 A\Rightarrow B
   &\longrightarrow \neg A\lor B,\\
 A\Leftrightarrow B
   &\longrightarrow (\neg A\lor B)\land(\neg B\lor A),\\
 \operatorname{ite}(C,A,B)
   &\longrightarrow (C\land A)\lor(\neg C\land B).
\end{align*}
The same Boolean equations imply
\begin{align*}
 \neg(A\Rightarrow B)&\longrightarrow A\land\neg B,\\
 \neg(A\Leftrightarrow B)&\longrightarrow
   (A\land\neg B)\lor(B\land\neg A).
\end{align*}
These five \textsc{rule-005} equalities are proved only for the formal Boolean
truth functions.  Classifying and rewriting formula-valued versus
expression-valued Alloy conditionals is an artifact-refinement obligation.

\leanref{TSG-RULE-006}{TypedSlottedEGraphsPaper.ProfileRules.tsgRule006_booleanNegationNormalForm}
The Boolean NNF rules are
\begin{align*}
 \neg\top&\longrightarrow\bot,&
 \neg\bot&\longrightarrow\top,&
 \neg\neg A&\longrightarrow A,\\
 \neg\bigwedge_i A_i&\longrightarrow\bigvee_i\neg A_i,&
 \neg\bigvee_i A_i&\longrightarrow\bigwedge_i\neg A_i.
\end{align*}
Here \textsc{rule-006} is exactly the two constant equations, double negation,
and the two De Morgan equations for finite lists of Booleans.

\leanref{TSG-RULE-007}{TypedSlottedEGraphsPaper.ProfileRules.tsgRule007_involutiveAtomicNegationDuals}
Atomic negation uses the following involutive duals:
\[
\begin{array}{c@{\quad\leftrightarrow\quad}c@{\qquad}c@{\quad\leftrightarrow\quad}c}
 = & \ne & > & \le \\
 \ge & < & \in & \notin \\
 \multicolumn{4}{c}{\operatorname{some}R\quad\leftrightarrow\quad\operatorname{no}R}
\end{array}
\]
The exact \textsc{rule-007} theorem is conditional: if the finite
$\operatorname{atomicDual}$ function returns $o^\bot$ for $o$, it returns $o$
for $o^\bot$, and normalizing the negation of $o$ returns the atom
$o^\bot$.  Behavior of parser-specific negative comparison heads and
unsupported operators is not part of this theorem; Java conformance to the
table is checked separately.

\subsection{Temporal Negation}

\leanref{TSG-RULE-008}{TypedSlottedEGraphsPaper.ProfileRules.tsgRule008_temporalNegationDuals}
The supported temporal duals are
\begin{align*}
 \neg\operatorname{always}A
   &\longrightarrow\operatorname{eventually}\neg A,\\
 \neg\operatorname{eventually}A
   &\longrightarrow\operatorname{always}\neg A,\\
 \neg\operatorname{historically}A
   &\longrightarrow\operatorname{once}\neg A,\\
 \neg\operatorname{once}A
   &\longrightarrow\operatorname{historically}\neg A,\\
 \neg(A\ \mathsf{until}\ B)
   &\longrightarrow(\neg A)\ \mathsf{releases}\ (\neg B),\\
 \neg(A\ \mathsf{releases}\ B)
   &\longrightarrow(\neg A)\ \mathsf{until}\ (\neg B),\\
 \neg(A\ \mathsf{since}\ B)
   &\longrightarrow(\neg A)\ \mathsf{triggered}\ (\neg B),\\
 \neg(A\ \mathsf{triggered}\ B)
   &\longrightarrow(\neg A)\ \mathsf{since}\ (\neg B).
\end{align*}
The exact \textsc{rule-008} theorem has the same conditional form: whenever
$\operatorname{temporalDual}(o)=o^\bot$, the reverse lookup returns $o$ and
normalization returns the dualized constructor $o^\bot$.  It does not prove a
temporal semantics, phase separation, or Java handling of \textsf{before} and
\textsf{after}.

\subsection{Quantifier Rewrites and Prenexing}

\leanref{TSG-RULE-009}{TypedSlottedEGraphsPaper.ProfileRules.tsgRule009_quantifierNegation}
Write $\mathcal Qx{:}D.P$ suggestively, but \textsc{rule-009} is formalized as
evaluation of a quantifier over a finite list of Boolean matrix values.
Negation follows exactly
\begin{align*}
 \neg(\forall x{:}D.P)&\longrightarrow\exists x{:}D.\neg P,\\
 \neg(\exists x{:}D.P)&\longrightarrow\forall x{:}D.\neg P,\\
 \neg(\mathsf{no}\ x{:}D.P)&\longrightarrow\exists x{:}D.P,\\
 \neg(\mathsf{one}\ x{:}D.P)&\longrightarrow\mathsf{notone}\ x{:}D.P,\\
 \neg(\mathsf{notone}\ x{:}D.P)&\longrightarrow\mathsf{one}\ x{:}D.P,\\
 \neg(\mathsf{lone}\ x{:}D.P)&\longrightarrow\mathsf{notlone}\ x{:}D.P,\\
 \neg(\mathsf{notlone}\ x{:}D.P)&\longrightarrow\mathsf{lone}\ x{:}D.P.
\end{align*}
The theorem is the corresponding Boolean-list equality for all seven formal
quantifier constructors; it is not an Alloy evaluator correspondence theorem.

\leanref{TSG-RULE-010}{TypedSlottedEGraphsPaper.ProfileRules.tsgRule010_emptyAdmissibleBindings}
For an empty domain, independently of $P$,
\[
\begin{array}{rcl@{\qquad}rcl}
 \forall x{:}\varnothing.P&\longrightarrow&\top &
 \exists x{:}\varnothing.P&\longrightarrow&\bot\\
 \mathsf{no}\ x{:}\varnothing.P&\longrightarrow&\top &
 \mathsf{one}\ x{:}\varnothing.P&\longrightarrow&\bot\\
 \mathsf{lone}\ x{:}\varnothing.P&\longrightarrow&\top &
 \mathsf{notone}\ x{:}\varnothing.P&\longrightarrow&\top\\
 \mathsf{notlone}\ x{:}\varnothing.P&\longrightarrow&\bot.&&
\end{array}
\]
For \textsc{rule-010}, ``empty'' means exactly that the formal admissible-value
list is $[]$; the seven displayed values are then proved by evaluation.
Determining emptiness from Alloy declaration syntax and multiplicity is outside
this theorem and is a required frontend/Java check.

\leanref{TSG-RULE-011}{TypedSlottedEGraphsPaper.ProfileRules.tsgRule011_constantBodyQuantifiers}
The carrier-independent constant-body rules are
\[
\begin{array}{rcl@{\qquad}rcl}
 \forall x{:}D.\top&\longrightarrow&\top &
 \exists x{:}D.\bot&\longrightarrow&\bot\\
 \mathsf{no}\ x{:}D.\bot&\longrightarrow&\top &
 \mathsf{one}\ x{:}D.\bot&\longrightarrow&\bot\\
 \mathsf{lone}\ x{:}D.\bot&\longrightarrow&\top &
 \mathsf{notone}\ x{:}D.\bot&\longrightarrow&\top\\
 \mathsf{notlone}\ x{:}D.\bot&\longrightarrow&\bot.&&
\end{array}
\]
These are exactly the seven \textsc{rule-011} equalities for a list of any
length.  No value is asserted for existential true or universal false, whose
Boolean-list values depend on whether that list is empty.

\leanref{TSG-RULE-012}{TypedSlottedEGraphsPaper.ProfileRules.tsgRule012_phaseLocalPrenexAndDomainGuarding}
The exact \textsc{rule-012} statement is four finite Boolean-list equations.
For a carrier list $C$, predicates $p,d:C\to\mathbb B$, and Boolean $r$,
\begin{align*}
 (\exists_C p)\land r&=\exists_C(\lambda x.\,p(x)\land r),&
 (\forall_C p)\lor r&=\forall_C(\lambda x.\,p(x)\lor r),\\
 \exists_{\operatorname{filter}(d,C)}p
   &=\exists_C(\lambda x.\,d(x)\land p(x)),&
 \forall_{\operatorname{filter}(d,C)}p
   &=\forall_C(\lambda x.\,\neg d(x)\lor p(x)).
\end{align*}
Its binding-tuple operation changes only the primitive carrier and preserves
the quantifier, lower and upper multiplicities, disjointness class, and phase
fields.  Fresh-variable conditions, general AST free-variable preservation,
cardinality lowering, sums, and comprehensions are not conclusions of this
claim.

\leanref{TSG-RULE-013}{TypedSlottedEGraphsPaper.ProfileRules.tsgRule013_descriptorCertifiedBlockPermutation}
For \textsc{rule-013}, an accepted finite permutation is packaged with a proof
that the supplied binder descriptor certifies it.  The theorem proves that it
occurs in the descriptor's supplied automorphism list and that, for every
environment coordinate $i$,
\[
  \mathit{env}\bigl(\pi^{-1}(\pi(i))\bigr)=\mathit{env}(i).
\]
It does not establish semantic preservation of a quantified Alloy body or
preservation of additional descriptor metadata; those are separate refinement
obligations.

\subsection{Flexible-Arity Structural Laws}

\leanref{TSG-RULE-014}{TypedSlottedEGraphsPaper.ProfileRules.tsgRule014_completeOperatorLawBoundary}
For \textsc{rule-014}, Table~\ref{tab:operator-laws} is exactly the finite
$\operatorname{operatorLaw}$ lookup table.  Its three Boolean columns are
policy flags, not semantic certificates.  In particular, the theorem proves
the noninteger lookup equations; the integer row belongs to
\textsc{rule-015}.  Neither theorem derives a source-level law merely from a
flag.

\begin{table}[t]
\centering
\footnotesize
\caption{Formal carrier assignments and policy flags.}
\label{tab:operator-laws}
\setlength{\tabcolsep}{4pt}
\begin{tabular}{p{0.25\linewidth}p{0.20\linewidth}p{0.22\linewidth}p{0.17\linewidth}}
\hline
Formal operator & Carrier & C/I flags & Flat flag\\
\hline
$\mathsf{boolAnd},\mathsf{boolOr}$ & $\mathsf{setPlus}$ & true/true & true\\
$\mathsf{relationalUnion}$,\newline $\mathsf{relationalIntersection}$ & $\mathsf{setPlus}$ & true/true & true\\
$\mathsf{integerAddition}$,\newline $\mathsf{integerMultiplication}$ & $\mathsf{bagPlus}$ & true/false & profile-dependent\\
$\mathsf{relationalJoin}$,\newline $\mathsf{relationalProduct}$ & $\mathsf{fixedBinary}\allowbreak\mathsf{Positional}$ & false/false & false\\
$\mathsf{equality},\mathsf{inequality}$,\newline $\mathsf{biconditional}$ & $\mathsf{fixedBinary}\allowbreak\mathsf{Commutative}$ & true/false & false\\
$\mathsf{disjoint}(k)$ & $\mathsf{nonflatBag}$ & true/false & false\\
$\mathsf{call}(g,k)$ & $\mathsf{positional}(k)$ & false/false & false\\
$\mathsf{totalOrder}(k)$ & $\mathsf{positional}(k)$ & false/false & false\\
\hline
\end{tabular}
\end{table}

The disjoint, call, and total-order equations are quantified over their natural
number arities (and, for calls, the callee identifier).  Whether the Java
frontend enforces source arities and whether a policy flag is backed by the
required endpoint certificates are separate artifact checks.  No additional
role-sensitive-list or dependent-chain row is covered by \textsc{rule-014}.

\leanref{TSG-RULE-015}{TypedSlottedEGraphsPaper.ProfileRules.tsgRule015_modularIntegerACProfile}
The exact \textsc{rule-015} model uses $\operatorname{Fin}(2^w)$.  Its addition
and multiplication are proved associative and commutative; the modular profile
has both flat flags set, and the overflow-forbidding profile has both cleared.
\leanref{TSG-CE-003}{TypedSlottedEGraphsPaper.ProfileRules.tsgCe003_overflowReassociationCounterexample}
Separately, the checked signed four-bit functions prove
$\operatorname{checkedAdd4}(7,1)=\mathsf{none}$,
$\operatorname{checkedAdd4}(1,-1)=\mathsf{some}(0)$, and hence the left
association fails while the right returns $7$.  Connecting either model to an
Alloy or Java integer mode is not part of these Lean results.

In the abstract \textsc{alg-009} procedure, sequence and bag constructors map
their child lists recursively, whereas a set constructor additionally applies
the supplied key-deduplication function.  This statement is about lists in the
formal quotient datatype; it does not prove multiset semantics, Java sorting,
or graph-invocation equality.

\subsection{Local Saturation Rules}

\leanref{TSG-RULE-016}{TypedSlottedEGraphsPaper.ProfileRules.tsgRule016_booleanLocalSaturation}
The Boolean identities, annihilators, and complements are
\begin{align*}
 A\land A&\longrightarrow A,& A\lor A&\longrightarrow A,\\
 A\land\top&\longrightarrow A,& A\lor\bot&\longrightarrow A,\\
 A\land\bot&\longrightarrow\bot,& A\lor\top&\longrightarrow\top,\\
 A\land\neg A&\longrightarrow\bot,&
 A\lor\neg A&\longrightarrow\top.
\end{align*}
These eight \textsc{rule-016} identities are Boolean equations.  The associated
formal smart constructor returns a constant on an empty input list, the child
on a singleton, and a nonempty-port record otherwise; its stored-port arity can
never be $0$.  No theorem here connects that datatype to a frontend
$\texttt{Set}^+$ node or to an emitted unit certificate.

\leanref{TSG-RULE-017}{TypedSlottedEGraphsPaper.ProfileRules.tsgRule017_atomicComplementSaturation}
For \textsc{rule-017}, an atomic literal consists only of an arbitrary atom and
a positive/negative polarity.  Complement flips the polarity while retaining
the same atom, and for any Boolean valuation
\[
  \ell\land\overline\ell=\bot,
  \qquad
  \ell\lor\overline\ell=\top.
\]
Recognizing concrete operator duals and equal canonical e-class invocations is
an implementation premise, not a conclusion of this theorem.

\leanref{TSG-RULE-018}{TypedSlottedEGraphsPaper.ProfileRules.tsgRule018_relationalLocalSaturation}
For \textsc{rule-018}, relations are abstract predicates
$\alpha\to\mathsf{Prop}$.  Given a witness that $x$ is nonempty, the theorem
proves pointwise
\begin{align*}
 x\in\mathsf{none}&\longrightarrow\bot,&
 x\in\mathsf{univ}&\longrightarrow\top,\\
 R+\mathsf{none}&\longrightarrow R,&
 R\mathbin{\&}\mathsf{none}&\longrightarrow\mathsf{none},\\
 R+R&\longrightarrow R,&
 R\mathbin{\&}R&\longrightarrow R.
\end{align*}
and also that the empty predicate is a subset of itself.  This theorem contains
no typed Alloy AST, multiplicity analysis, ACI-port inference, or Java rewrite
claim.

\leanref{TSG-RULE-019}{TypedSlottedEGraphsPaper.ProfileRules.tsgRule019_implicationLocalSaturation}
Finally, \textsc{rule-019} proves the Boolean identities
\begin{align*}
 \bot\Rightarrow A&\longrightarrow\top,&
 \top\Rightarrow A&\longrightarrow A,\\
 A\Rightarrow\top&\longrightarrow\top,&
 A\Rightarrow\bot&\longrightarrow\neg A,
\end{align*}
and the Boolean equality $A\Rightarrow B=\neg A\lor B$.  Stage placement and
Java saturation behavior are outside this equality theorem.

\subsection{Boundary and Implementation Correspondence}

\leanref{TSG-SCOPE-001}{TypedSlottedEGraphsPaper.ProfileRules.tsgScope001_zeroCanonicalDistanceBoundary}
The exact \textsc{scope-001} theorem is about an arbitrary complete finite
normalization presentation $P$.  Define its discrete canonical distance to be
$0$ when $P.\mathit{normal}(x)=P.\mathit{normal}(y)$ and $1$ otherwise.  Then
\[
  \operatorname{canonicalDistance}_P(x,y)=0
  \quad\Longleftrightarrow\quad P.\mathit{relation}(x,y).
\]
This follows from the presentation's normal-form exactness.  It neither proves
that the individual rules in this appendix generate one combined $P$, nor
identifies this $0/1$ function with the dataset's AST edit distance, nor states
a Java-refinement premise or conclusion.

The remaining production statements are artifact-only obligations.  In
particular, graph-relative typed-renaming enumeration, binder-group closure,
the behavior of \texttt{SlotPermutationGroup}, the 720-permutation guard, and
the 32-pass rewrite bound are not Lean conclusions.  The publication run gives
in-process checked observations under its recorded rule set; it does not prove
those correspondences and cannot be cited as realizing \textsc{scope-001} or
the abstract algorithm contracts above.

\end{document}